\documentclass[sigconf, nonacm]{acmart}
\usepackage{amsmath,amsfonts}
\usepackage{mathtools}
\usepackage{tikz}
\usetikzlibrary{arrows,shapes,fadings,positioning,patterns,decorations,calc,decorations.pathreplacing,calligraphy}
\usepackage{semantic}
\usepackage{xspace}

\theoremstyle{acmplain}
\newtheorem{theorem}{Theorem}
\newtheorem{proposition}[theorem]{Proposition}
\newtheorem{lemma}[theorem]{Lemma}

\newtheorem{remark}[theorem]{Remark}

\theoremstyle{acmdefinition}
\newtheorem{definition}[theorem]{Definition}

\usepackage{algorithm} 
\usepackage[noend]{algpseudocode}

\usepackage[shortlabels]{enumitem}
\setlist{labelindent=\parindent, leftmargin=*}

\usepackage{subcaption}

\usepackage{framed}

\newcommand{\OMIT}[1]{}

\usepackage{listings}

\newcommand{\gpml}{\text{GPC}\xspace}
\newcommand{\gpmlplus}{\text{GPC+}\xspace}

\usepackage{bold-extra}

\newcommand{\collect}{\mathbf{collect}}

\newcommand{\pathlen}{\textsf{len}}
\newcommand{\pathval}{\textsf{path}}
\newcommand{\listval}{\textsf{list}}

\def\endpoints{\textsf{endpoints}}

\newcommand{\concat}{\mathbin{\cdot}}

\makeatletter
\newcommand{\set}[2][]{{#1\{}#2{#1\}}}
\newcommand{\setst}{\@ifstar{\autosetst}{\paramsetst}}
\newcommand{\autosetst}[2]{\left\lbrace\,#1~\middle|~#2\,\right\rbrace}
\newcommand{\paramsetst}[3][]{#1\{\,#2\mathbin{#1|}#3\,#1\}}
\newcommand{\bag}{\@ifstar{\autobag}{\parambag}}
\newcommand{\autobag}[1]{\llbrace*#1\rrbrace*}
\newcommand{\parambag}[2][]{\llbrace[#1]#2\rrbrace[#1]}
\newcommand{\bagst}{\@ifstar{\autobagst}{\parambagst}}
\newcommand{\autobagst}[2]{\llbrace*\,#1\,\middle|\,#2\,\rrbrace*}
\newcommand{\parambagst}[3][]{\llbrace[#1]\,#2\mathbin{#1|}#3\,\rrbrace[#1]}

\newcommand{\llbrace}{\@ifstar{\leftllbrace}{\paramllbrace}}
\newcommand{\paramllbrace}[1][]{{#1\{\hspace*{-.25em}#1\{}}
\newcommand{\leftllbrace}{\left\lbrace\kern-3\nulldelimiterspace\middle\lbrace}
\newcommand{\rrbrace}{\@ifstar{\rightrrbrace}{\paramrrbrace}}
\newcommand{\paramrrbrace}[1][]{{#1\}}\hspace*{-.25em}{#1\}}}
\newcommand{\rightrrbrace}{\middle\rbrace\kern-3\nulldelimiterspace\right\rbrace}
\makeatother

\newcommand{\gqlil}[1]{\lstinline`#1`}

\newcommand{\node}{\textsf{node}}

\newcommand{\dom}{\textsf{dom}}
\newcommand{\card}{\textsf{card}}
\newcommand{\var}{\textsf{var}}
\newcommand{\sch}{\textsf{sch}}

\newcommand{\sem}[1]{\left\llbracket#1\right\rrbracket}

\newcommand{\operator}{\mathbin{\Theta}}

\tikzset{>=stealth}
\tikzset{
  detail/.style={%
    draw, rectangle split, rectangle split parts=#1,
    rectangle split draw splits=true
  },
  onslide/.code args={<#1>#2}{\only<#1>{\pgfkeysalso{#2}}},
  invisible/.style={opacity=0},
  visible on/.style={alt=#1{}{invisible}},
  alt/.code args={<#1>#2#3}{%
      \alt<#1>{\pgfkeysalso{#2}}{\pgfkeysalso{#3}} 
  },
}
\tikzset{%
  nodeid/.style={%
    circle, draw, minimum size=4mm, line width=0.5mm, inner sep=1mm, outer sep=0.75mm,
  },
  edgeid/.style={%
    line width=0.5mm, circle, minimum size=5mm, draw=none, fill=white, inner sep=0mm, outer sep=1mm,
  },
  relationship/.style={%
    draw, very thick, -latex
  },
  existing/.style={%
    draw=gray,
  },
  edgelabel/.style={%
    pos=.433,
    above,
  }
}

\makeatletter
\tikzset{
    dot diameter/.store in=\dot@diameter,
    dot diameter=3pt,
    dot spacing/.store in=\dot@spacing,
    dot spacing=10pt,
    dots/.style={
        line width=\dot@diameter,
        line cap=round,
        dash pattern=on 0pt off \dot@spacing
    }
}
\makeatother

\usepackage{color}
\definecolor{gray}{rgb}{0.4,0.4,0.4}
\definecolor{darkblue}{rgb}{0.0,0.0,0.6}
\definecolor{darkred}{rgb}{0.45,0,0}
\definecolor{darkgreen}{rgb}{0,0.30,0.20}
\definecolor{darkpurple}{RGB}{120, 0, 180}
\colorlet{keywordcolor}{darkblue}
\newcommand{\kwfont}{\color{keywordcolor}\bfseries}

\colorlet{labelcolor}{darkgreen}
\newcommand{\lblfont}{\color{labelcolor}}

\colorlet{keycolor}{darkred}
\newcommand{\keyfont}{\color{keycolor}}

\colorlet{structcolor}{black}
\newcommand{\structfont}{\color{structcolor}\bfseries}

\usepackage{listings,textcomp}
\lstdefinelanguage{cypher}
{
  columns=fullflexible,
  otherkeywords={<,>,-,[,],\{,\},|,:,?,*},
  keywords=[1]{},
  keywordstyle=,
  keywordstyle=[1]\color{darkpurple},
  keywords=[2]{CONTINUE, ACYCLIC, SIMPLE, MATCH,KEEP,WHERE,WITH,OPTIONAL,RETURN,MERGE,CREATE,SET,DETACH,DELETE,REMOVE,ALL,ANY,GROUP,SAME,IS,NOT,AND,THEN,OR,YIELDS,DISTINCT,PATHS,SHORTEST,CHEAPEST,TRAIL,WALK,COST},
  keywordstyle=[2]\kwfont,
  keywords=[3]{<,>,-,[,],(,)},
  keywordstyle=[3]\structfont,
  keywords=[5]{Woman,L1,L2,L3,L3},
  keywordstyle=[5]\lblfont,
  keywords=[6]{BOUGHT_WITH,MOVIE,FRIEND,PARTNER,BOOK,TRANSFERS,CORP,BANK,SPOUSE,PERSON},
  keywordstyle=[6]\lblfont,
  keywords=[7]{name,id,genre},
  keywordstyle=[7]\keyfont,
  string=[m]{"},
  stringstyle=,
}
\lstnewenvironment{cypher}[1][]{\lstset{language=cypher,#1}}{}

\newcommand{\forcedhfill}{\hspace*{0pt plus 1fill}}
\makeatletter
\newcounter{query}
\renewcommand{\thequery}{(\arabic{query})}
\newcounter{nextquery}

\newlength{\lstskip}
\newcommand{\vm@ccypher@start}[2][]{%
  \leavevmode\unskip\pagebreak[1]\vspace{\lstskip}\par\noindent%
  #1%
  \lst@boxtrue%
  \lstset{language=cypher,boxpos=b,resetmargins=true,mathescape=true,#2}
  \forcedhfill}
\newcommand{\vm@ccypher@end}[1][]{\forcedhfill#1\par\addvspace{\lstskip}}
\lstnewenvironment{gql}[1][]%
  {\vm@ccypher@start[\refstepcounter{query}]{#1}}%
  {\vm@ccypher@end[~\thequery]}
\lstnewenvironment{gql*}[1][]%
  {\vm@ccypher@start{#1}}
  {\vm@ccypher@end}
\makeatother

\newcommand{\src}{\textsf{src}}
\newcommand{\tgt}{\textsf{tgt}}
\newcommand{\dv}{\delta}
\newcommand{\LL}{\mathcal{L}}

\newcommand{\arclit}[1]{%
  \if\relax\detokenize{#1}\relax
  \mathrel{\smash{\xLeftrightarrow{}}}
  \else
  \mathrel{\xLeftrightarrow{\,#1\,}}
  \fi
}
\newcommand{\rightlit}[1]{%
  \if\relax\detokenize{#1}\relax
  \mathrel{\smash{\xrightarrow{}}}
  \else
  \mathrel{\xrightarrow{\,#1\,}}
  \fi
}
\newcommand{\leftlit}[1]{%
  \if\relax\detokenize{#1}\relax
  \mathrel{\smash{\xleftarrow{}}}
  \else
  \mathrel{\xleftarrow{\,#1\,}}
  \fi
}
\newcommand{\shortundirlit}{-\hspace{-3pt}-}

\newcommand{\undirlit}[2][-\hspace{-3pt}-\hspace{-3pt}-]{%
  \if\relax\detokenize{#2}\relax
  \mathrel{#1}
  \else
  \stackrel{#2}{#1}
  \fi
}

\newcommand{\interval}[2]{{#1}..{#2}}

\newcommand{\nodelit}[1]{({#1})}
\newcommand{\quantlit}[1]{\mathpunct{{}^{#1}}}
\newcommand{\condlit}[1]{{}_{\langle #1\rangle}}
\newcommand{\simplelit}{\mathsf{simple}}
\newcommand{\traillit}{\mathsf{trail}}

\newcommand{\shortest}{\mathsf{shortest}}

\newcommand{\descr}{d}

\newcommand{\pat}{\pi}

\newcommand{\restrictor}{\rho}
\newcommand{\rsimple}{\simplelit}
\newcommand{\rtrail}{\traillit}

\newcommand{\cnt}[1]{\sharp(#1)}

\newcommand{\NN}{\ensuremath{\mathcal{N}}}
\newcommand{\EE}{\ensuremath{\mathcal{E}}}
\newcommand{\VV}{\ensuremath{\mathcal{V}}}
\newcommand{\KK}{\ensuremath{\mathcal{K}}}
\newcommand{\XX}{\ensuremath{\mathcal{X}}}
\newcommand{\PP}{\ensuremath{\mathsf{Paths}}}
\newcommand{\TT}{\ensuremath{\mathcal{T}}}

\newcommand{\VVtype}[1]{\VV_{#1}}

\newcommand{\const}{\ensuremath{\textsf{Const}}}

\def\stoptoken{!}
\newif\ifpathstop
\def\vmedge#1,{\if\stoptoken#1\let\vmnode\relax\else,\textcolor{red}{#1},\fi\vmnode}
\def\vmnode#1,{\if\stoptoken#1\let\vmedge\relax\else\textcolor{blue}{#1}\fi\vmedge}%

\definecolor{darkblue}{RGB}{0,56,153}
\definecolor{darkred}{RGB}{153,0,0}

\newcommand{\Paths}{\mathsf{Paths}}

\newcommand{\tnode}{\mathsf{Node}}
\newcommand{\tedge}{\mathsf{Edge}}

\newcommand{\tbool}{\mathsf{Bool}}

\newcommand{\tpath}{\mathsf{Path}}

\newcommand{\tlist}{\mathsf{Group}}

\newcommand{\tmaybe}{\mathsf{Maybe}}
\newcommand{\tnothing}{\mathsf{Nothing}}

\newcommand{\np}{\textsf{NP}\xspace}

\newcommand{\semoneline}[2]{%
#1={}&#2
\displaybreak[0]%
}
\newcommand{\semtwoline}[3][1em]{%
#2={}&\\
\multispan2{\hspace*{#1 plus 1fil}$\displaystyle #3$}%
\\[-8pt]
\displaybreak[0]%
}

\let\colondash\relax
\DeclareMathOperator{\colondash}{\mbox{:--}}

\newcommand{\ans}{\mathit{Ans}}

\newcommand{\expr}{\xi}

\title{\gpml: A Pattern Calculus for Property Graphs}

\author{Nadime Francis}
\email{nadime.francis@univ-eiffel.fr}
\affiliation{%
    \institution{LIGM, U Gustave Eiffel, CNRS}%
    \city{Champs-sur-Marne}
    \country{France}
}

\author{Am{\'e}lie Gheerbrant}
\email{Amelie.Gheerbrant@irif.fr}
\affiliation{%
    \institution{IRIF, U Paris Cit{\'e}, CNRS}%
    \city{Paris}
    \country{France}
}

\author{Paolo Guagliardo}
\email{Paolo.Guagliardo@ed.ac.uk}
\affiliation{\institution{U Edinburgh}\country{UK}}

\author{Leonid Libkin}
\email{l@libk.in}
\affiliation{\institution{U Edinburgh \& RelationalAI \& ENS}\country{UK \& France}}

\author{Victor Marsault}
\email{victor.marsault@univ-eiffel.fr}
\affiliation{%
    \institution{LIGM, U Gustave Eiffel, CNRS}%
    \city{Champs-sur-Marne}
    \country{France}
}

\author{Wim Martens}
\email{wim.martens@uni-bayreuth.de}
\affiliation{\institution{University of Bayreuth}\country{Germany}}

\author{Filip Murlak}
\email{f.murlak@uw.edu.pl}
\affiliation{\institution{University of Warsaw}\country{Poland}}

\author{Liat Peterfreund}
\email{liat.peterfreund@univ-eiffel.fr}
\affiliation{%
    \institution{LIGM, U Gustave Eiffel, CNRS}%
    \city{Champs-sur-Marne}
    \country{France}
}

\author{Alexandra Rogova}
\email{rogova@irif.fr}
\affiliation{%
    \institution{IRIF, U Paris Cit{\'e}, CNRS}%
    \country{}
}
\affiliation{
    \institution{diiP, Inria}    
    \city{Paris}
    \country{France}
}

\author{Domagoj Vrgo\v{c}}
\email{domagojvrgoc@gmail.com}
\affiliation{\institution{PUC Chile \& IMFD}\country{Chile}}

\date{\today}

\usepackage{hyperref}
\begin{document}

\begin{abstract}
The development of practical query languages for graph databases runs well ahead of the underlying theory. The ISO committee in charge of database query languages is currently developing a new standard called {\em Graph Query Language} (GQL) as well as an extension of the SQL Standard for querying property graphs represented by a relational schema, called SQL/PGQ. The main component of both is the pattern matching facility, which is shared by the two standards. In many aspects, it goes well beyond RPQs, CRPQs, and similar queries on which the research community has focused for years.

Our main contribution is to distill the lengthy standard specification into a simple Graph Pattern Calculus (\gpml) that reflects all the key pattern matching features of GQL and SQL/PGQ, and at the same time lends itself to rigorous theoretical investigation.
We describe the syntax and semantics of \gpml, along with the typing rules that ensure its expressions are well-defined, and state some basic properties of the language. With this paper we provide the community a tool to embark on a study of query languages that will soon be widely adopted by industry.
\end{abstract}

\maketitle

\section{Introduction}
\label{sec:intro}

The foundations of graph databases were laid more than 30 years ago in papers that define the now ubiquitous notion of regular path queries (RPQs) \cite{RPQ,CRPQ}. They preceded the development of graph database systems by decades: it was in this millennium that the graph database industry properly emerged, driven by two similar models, namely \emph{RDF data} and {\em property graphs}. The latter one is now promoted by multiple vendors such as Oracle, Neo4j, Amazon, SAP, Redis, TigerGraph etc. A very notable development in the evolution of property graph databases is the decision, taken 3 years ago, to produce a new standard query language called GQL \cite{GQLwiki}. It would subsume the hitherto used languages such as Cypher \cite{Cypher} of Neo4j, PGQL \cite{PGQL} of Oracle, GSQL \cite{tigergraph-sigmod} of TigerGraph, and the G-CORE proposal from an industry/academia group \cite{gcore}.

The new GQL (Graph Query Language) Standard is developed by the same ISO committee that is in charge of developing and maintaining the SQL Standard. The core of any graph query language is its {\em pattern matching} engine, that finds patterns in graphs. GQL's pattern matching facilities are in fact shared across two standards:
\begin{itemize}
    \item SQL/PGQ, a new Part 16 of the SQL standard, that defines querying graphs specified as views over a relational schema;
    \item GQL, a standalone language for querying property graphs.
\end{itemize}

The development of GQL as a query language standard is rather different from SQL. The latter came out of a well-researched relational theory;  relational calculus led to its declarative approach, while relational algebra formed the foundation of RDBMS implementations. But while GQL is "inspired" by the key developments of database research,  they are not represented directly in the language, which itself is designed by an industry consortium. The GQL committee lists\footnote{\url{https://www.gqlstandards.org/existing-languages}} three main academic influences: regular path queries~\cite{RPQ}, Graph XPath \cite{gxpath-jacm}, and regular queries on graphs \cite{RRV17} which are the regular closure of conjunctive RPQs. While these provided important initial orientation, the GQL development is much more in line with industry-level languages such as Cypher, GSQL, and PGQL. 

The theory of graph query languages on the other hand produced a multitude of languages based on RPQs: CRPQs, UCRPQs, 2(UC)RPQs, ECRPQs, just to name a few (see \cite{surveyChile,B13} for many more). However, none of them can play the role of relational calculus with respect to the development of GQL, as they do not capture its key features with respect to both navigation and handling data.

{\bf Our goal} then is to produce that missing piece, a theoretical language that underlies GQL and SQL/PGQ pattern matching and can be studied in the same way as the RPQ family has been studied over decades. Doing so has two significant difficulties:
\begin{enumerate}
    \item GQL pattern matching is described by close to 100 pages of the Standard text which is not human-friendly; it is intended for developers implementing the language;
    \item Even that text is not available to the research community: the relevant ISO standards will only be published in 2023 and even then will be behind a paywall. 
\end{enumerate}

The only publicly available source is \cite{sigmod22}. It outlines the main features of GQL and SQL/PGQ pattern matching by means of several examples and thus hardly serves the purpose of introducing a pattern matching calculus to form the basis of further study of graph querying.
Such a calculus should judiciously choose the key features leaving others for extensions. Think again about the SQL/relational calculus analogy. The latter does not have bag semantics, nulls, aggregates, full typing, and many other SQL features. We adopt a similar approach here. Our goal is to introduce a calculus that captures the {\em essence} of GQL pattern matching, without every single feature present. Similarly to relational calculus, we settle for the core set-semantics fragment without nulls or aggregates, each of which can be added as extensions. 

\paragraph{What is new?}
We now outline the main differences between the currently existing theoretical languages and real-life pattern matching in property graphs that we need to capture. 

\begin{itemize} 
\item Much of the theoretical literature relies on an overly simplified model of graph databases as edge-labeled graphs; this is common in the study of RPQs and extensions \cite{B13}. Typical patterns in languages like Cypher combine graph navigation with querying data held in nodes and edges. With data values, the theoretical model of choice is data graphs \cite{gxpath-jacm}: in those, each node carries a single data value, similarly to data trees studied extensively in connection with tree-structured data \cite{BDMSS06}. Property graphs are yet more complex, as each edge or node can carry an {\em arbitrary collection of key-value pairs}. 

\item GQL patterns bind variables -- in different ways -- and use them to select patterns. As an example, consider a pattern (in our notation, to be introduced in Section \ref{sec:calculus}) $\nodelit{x:a} \rightlit{e:b} \nodelit{y:a}$ that looks for $b$-labeled edges between two $a$-labeled elements. While producing a match, it {\em binds} variables $x$ and $y$ to the start and end-nodes of the edge, and $e$ to the edge itself. 
If, on the other hand, we look at $\nodelit{x:a} \rightlit{e:b}\!\quantlit{1..\infty} \nodelit{y:a}$, then we match paths of length 1 or more (indicated by $1..\infty$) from $x$ to $y$. In this case $e$ gets bound not to an edge, but rather to a {\em list of edges}. This has immediate implications on conditions in which such a binding can be used.

\item Property graphs are multigraphs (there can be multiple edges between two endpoint vertices), pseudographs (there can be an edge looping from a vertex to itself), and mixed, or partially directed graphs (as an edge can be directed or undirected). This means edges can be traversed in different directions, or traversals can be indicated as having no direction at all. 

\item To ensure the number of paths returned is finite, real-life languages 
put additional restrictions on paths, such as insisting that they be trails (no repeated edges, as in Cypher), simple (no repeated nodes), or shortest (as in G-Core \cite{gcore}). 
Typically in research literature one considers each one of those semantics separately, but GQL permits mixing them. 

\item Similarly to Cypher, in GQL one can match paths and output them. In theoretical languages this feature is rather an exception \cite{gcore,BarceloLLW-tods12}.

\item Finally, one can apply conditions to filter matched paths. For example, after matching $\nodelit{x:a} \rightlit{e:b}\quantlit{\!1..\infty} \nodelit{y:a}$, one can apply a condition $x.k=y.k$ stating that property $k$ of both $x$ and $y$ is the same. Notice that we cannot talk similarly about properties of $e$, as they are bound to a list.
\end{itemize}

Other features of patterns are those in 
regular languages: concatenation, disjunction, repetition; they can be applied on top of already existing patterns, similarly to \cite{RRV17}.

Our main contribution is the {\em Graph Pattern-matching Calculus \gpml} that captures all the key features of GQL and SQL/PGQ pattern matching. Its syntax is described in Section \ref{sec:calculus} (after the definition of property graph concepts in Section \ref{sec:prelim}). To ensure the well-definedness of its expressions, the calculus comes with a {\em type system}, given in Section \ref{sec:types}. We give a formal semantics of \gpml\ in Section \ref{sec:semantics}, and provide a number of basic results on the complexity of the language, and its relationship with classical theoretical formalisms in Section~\ref{sec:complexity}. In Section \ref{sec:extensions} we outline some possible extensions and describe two concrete examples where theoretical studies of the language had a direct impact on the Standard as it was being written.

\section{Data Model}
\label{sec:prelim}

We take the standard (as currently adopted by the GQL Standard  committee \cite{LDBC:TR:TR-2021-01}) definition of property graphs. We assume disjoint countable sets $\NN, \EE_\mathsf{d}, \EE_\mathsf{u}$ of node, directed, and undirected edge ids, $\LL$ of labels, $\KK$ of keys, and $\const$ of constants. A property graph 
is a tuple 
$G = \langle N, E_\mathsf{d}, E_\mathsf{u}, \lambda, \endpoints, \src, \tgt, \dv\rangle$
where
\begin{itemize}
\item $N \subset \NN$ is a finite set of node ids used in $G$;
\item $E_\mathsf{d}\subset \EE_\mathsf{d}$ is a finite set of directed edge ids used in $G$;
\item $E_\mathsf{u}\subset \EE_\mathsf{u}$ is a finite set of undirected edge ids used in $G$;
\item $\lambda:  N  \cup  E_\mathsf{d} \cup E_\mathsf{u} \to 2^{\LL}$  
is a labeling function that associates with every id a (possibly empty) finite set
 of labels from $\LL$;
\item $\src, \tgt: E_\mathsf{d} \to N$ define source and target of a directed edge;
\item $\endpoints: E_\mathsf{u} \to 2^N$ so that $|\endpoints(e)|$ is 1 or 2 define endpoints of an undirected edge;
\item $\dv: (N  \cup  E_\mathsf{d} \cup E_\mathsf{u}) \times \KK \to \const$ is a partial function that associates a constant with an id and a key from $\KK$. 
\end{itemize}

We use \emph{node} and \emph{edge} to refer to node ids and edge ids, respectively, and call a node $u$ an \emph{$\ell$-node} iff $\ell \in \lambda(u)$; similarly for edges.  

A \emph{path} is an alternating sequence of nodes and edges that starts and ends with a node, that is, it is a sequence of the form 
\[u_0 e_1 u_1 e_2 \cdots e_{n} u_n\;, \]
where $u_0,\ldots,u_n$ are nodes and $e_1, \ldots,e_n$ are (directed or undirected) edges. Note that we allow $n=0$, in which case the path consists of a single vertex and no edges. For a path $p$ we denote $u_0$ as $\src(p)$ and $u_n$ as $\tgt(p)$; we also refer to $u_0$ and $u_n$ as the path's \emph{endpoints}.
The \emph{length} of a path $p$, denoted $\pathlen(p)$, is $n$, i.e., the number of occurrences of edge ids in $p$. We also use the term \emph{edgeless path} to refer to a path of length zero. We spell paths explicitly as $\pathval(u_0,e_1,u_1,\cdots, e_n,u_n)$. 
We denote the set of all paths by~$\PP$.

A \emph{path in $G$} is a path such that each edge in it connects the nodes before and after it in the sequence.\footnote{As is usual in the graph database literature \cite{openCypher,MendelzonW95,Woo,B13}, we use the term path to denote what is called \emph{walk} in the graph theory literature \cite{bollobas2013modern}.} More formally, it is a path $\pathval(u_0, e_1, u_1, e_2, \ldots, e_{n}, u_n)$
such that at least one of the following holds for each $i \in [n]$:
\begin{enumerate}[(a)]
    \item $\src(e_i) = u_{i-1}$ and $\tgt(e_i) = u_i$ in which case we speak of $e_i$ as a {\em forward} edge in the path;
    \item $\src(e_i) = u_i$ and $\tgt(e_i) = u_{i-1}$ in which case we speak of $e_i$ as a {\em backward} edge in the path;
    \item $\endpoints(e_i) = \{u_{i-1},u_i\}$ in which case we speak of $e_i$ as an {\em undirected} edge in the path.
\end{enumerate}
Here, both (a) and (b) can be true at the same time in the case of a directed self-loop. 
By $\Paths(G)$ we denote the set of paths in $G$. Notice that $\Paths(G)$ can be infinite. 

Two paths $p = \pathval(u_0,e_0,\ldots,u_k)$ and $p' =
\pathval(u'_0,e'_0,\ldots,u'_j)$ 
\emph{concatenate}
if $u_k = u'_0$, in
which case their \emph{concatenation} $p \concat p'$ is defined as
$\pathval(u_0,e_0,\ldots,u_k,e'_0,\ldots,u'_j)$. %
Note that if one of the paths consists of a single node, then it is a unit of
concatenation and does not change the result. That is, $p \concat \pathval(u)$
is defined iff $u=u_k$, in which case it equals $p$; likewise for $\pathval(u)
\concat p$ and $u=u_0$.

\section{Pattern calculus}
\label{sec:calculus}

\begin{figure}[t!]
  \newlength{\lindent}
  \setlength{\lindent}{6em}
  \fbox{%
    \begin{subfigure}[t]{\linewidth-2\fboxsep-2\fboxrule}
      \begin{flushleft}
        \makebox[\lindent][l]{\textbf{BASICS}}
        For $x,y \in \XX$, $\ell \in \LL$, $a,b \in \KK$, $c\in \const$:
      \end{flushleft}
      \centering
      \begin{minipage}{.8\linewidth}
        \begin{flalign*}
          &(\text{descriptor}) &
          \descr ~\Coloneqq{}~& x ~\mid~ {:\ell} ~\mid~ x:\ell
          \\
          &(\text{direction}) &
          {\arclit{\hphantom{d}}} ~\Coloneqq{}~& {\rightlit{}} %
          ~\mid~ {\leftlit{}} ~\mid~ {\shortundirlit}
          \\
          &(\text{condition}) &
          \theta ~\Coloneqq{}~& x.a=c %
          ~\mid~ x.a=y.b \\
          && ~\mid{}~& \theta\land\theta %
          ~\mid~ \theta\lor\theta %
          ~\mid~ \lnot\theta %
          \\
          & (\text{restrictor}) &
          \restrictor ~\Coloneqq{}~& \rsimple ~\mid~ \rtrail ~\mid~ \shortest \\
          && ~\mid{}~& \shortest\,\rsimple ~\mid~ \shortest\,\rtrail
        \end{flalign*}
      \end{minipage}
    \end{subfigure}%
  }%
  \\[1ex]
  \fbox{%
    \begin{subfigure}[t]{\linewidth-2\fboxsep-2\fboxrule}
      \begin{flushleft}
        \makebox[\lindent][l]{\textbf{PATTERNS}} For $0 \leq n \leq m \leq \infty$:
      \end{flushleft}
      \centering
      \begin{minipage}{.65\linewidth}
        \begin{flalign*}
          \pat ~\Coloneqq{}~&
          \nodelit{\,} ~\mid~ \nodelit{\descr} &(\text{node pattern})&
          \\[-3pt]
          ~\mid{}~& {\arclit{\hphantom{d}}} ~\mid~ {\arclit{d}} &(\text{edge pattern})&
          \\
          ~\mid{}~& \pat + \pat & (\text{union})&
          \\
          ~\mid{}~& \pat\,\pat & (\text{concatenation})&
          \\
          ~\mid{}~& \pat\condlit{\theta} & (\text{conditioning})&
          \\
          ~\mid{}~& \pat\quantlit{{\interval{n}{m}}} & (\text{repetition})&
        \end{flalign*}
      \end{minipage}
    \end{subfigure}%
  }
  \\[1ex]
  \fbox{%
    \begin{subfigure}[t]{\linewidth-2\fboxsep-2\fboxrule}
      \begin{flushleft}
        \makebox[\lindent][l]{\textbf{QUERIES}} For $x \in \XX$:
      \end{flushleft}
      \centering
      \begin{minipage}{.65\linewidth}
        \begin{flalign*}
          Q ~\Coloneqq{}~& \rho\,\pi ~\mid~ x=\rho\,\pi & (\text{pattern})&
          \\
          ~\mid{}~& Q,Q & (\text{join})&
        \end{flalign*}
      \end{minipage}
    \end{subfigure}%
  }
  \caption{\gpml expressions}
  \label{f:patterns-fig}
\end{figure}

We assume a countably infinite set $\XX$ of \emph{variables}. The basic building
blocks of \gpml are {\em node patterns} and {\em edge} (or {\em arrow}) {\em
  patterns}. 
  Node patterns are of the form $\nodelit{x:\ell}$. Here $x$ is a variable, and $:\ell$ specifies the node label. 
  The brackets ``('' and ``)'' are mandatory, and signify that we are talking about a node.  Both the variable $x$, and the label specification $:\ell$ are optional, and can be omitted. This way, the simplest node pattern is $\nodelit{}$, matching any node in the graph.
The presence of a
variable means that it gets bound; the presence of a label $\ell$ means that
only $\ell$-nodes are matched. Edge patterns are of the form
$\arclit{x:\ell}$, where again $x$ and $\ell$ are a  variable and
an edge label, respectively, and $\arclit{}$ is one of the allowed
directions: $\rightlit{}$ (forward), $\leftlit{}$ (backward), ${\shortundirlit}$
(undirected). Both $x$ and $:\ell$ can be omitted. In the case when they are present, the variable $x$ gets bound to the matching edge, and $:\ell$ constrains the allowed edge labels.

The full grammar of \gpml is given in Fig.~\ref{f:patterns-fig}. Here:
\begin{itemize}
\item[$\descr$] specifies node and edge {\em descriptors}; these may include a
  variable to which that graph element is bound, and its label.
\item[$\arclit{\hphantom{d}}$] specifies possible edge {\em directions}: forward, backward,
  and un\-directed.
\item[$\theta$] defines {\em conditions}: atomic ones compare property values held in
  nodes or edges to one another or to constants, and conditions are closed under
  Boolean connectives.
\item[$\restrictor$] specifies {\em restrictors} on paths to ensure a finite
  result set; paths can be restricted to be \emph{simple} (no repeated nodes),
  \emph{trail} (no repeated edges), or \emph{shortest}, which can be optionally
  combined with simple or trail.
\item[$\pi$] defines {\em patterns}: the atomic ones are node and edge patterns,
  which have an optional descriptor and, for the latter, a mandatory direction;
  patterns are then built from these using {\em concatenation} (denoted by
  juxtaposition), {\em union} ($+$), {\em conditioning} (akin to selection in
  relational algebra), and {\em repetition} of the form $\interval{n}{m}$,
  meaning that the pattern is repeated between $n$ and $m$ times. Note that the
  $\interval{0}{\infty}$ repetition is precisely the Kleene star.
\item[$Q$] defines a {\em query}: a non-empty list of optionally named
  ($x=\restrictor\pat$) path patterns, each qualified by a restrictor.
\end{itemize}
When we write patterns, we disambiguate with square brackets, and the lower operator takes precedence. For instance, $\pi\pi'\condlit{\theta}+\pi'' = [\pi[\pi'\condlit{\theta}]]+\pi''$.
By an \emph{expression} of \gpml we shall mean a pattern or a query. 

\paragraph{Examples}

The formal semantics is presented in Section \ref{sec:semantics}. Next we
illustrate how \gpml operates with several examples. Each path pattern is matched to a path;
such a path could be a single node, an edge (with endpoints included), or a more
complex path.  For example, consider the pattern
\begin{equation*}
  \nodelit{x_1:A} %
  \rightlit{y_1} \nodelit{x_2:B} %
  \leftlit{y_2}  \nodelit{x_3:C} %
  \rightlit{y_3} \nodelit{x_1} 
\end{equation*}
matches a path from an $A$-node to itself via $B$- and $C$-nodes with the first
and third edges going forward and the second edge going backward.
Notice that this pattern introduces an implicit join over the endpoints of the path by repeating the variable $x_1$.

The pattern $\nodelit{x:A}\rightlit{}\nodelit{z:B} \, \big[\!\leftlit{}\nodelit{u:C}  + \nodelit{\,} \big]$ is an {\em optional} pattern\footnote{Note that we use square brackets for grouping, since $\nodelit{}$ defines a node pattern.}: a feature present in many languages such as SPARQL and Cypher. It matches an edge from an $A$-node to a $B$-node, binding $x$ and $z$ to its endpoints, and if the $B$-node has an incoming edge from a $C$-node, binds $u$ to its source. This pattern can be seen as the disjunction $\pi_1+\pi_2$ where $\pi_1=\nodelit{x:A}\rightlit{}\nodelit{z:B} \leftlit{}\nodelit{u:C}$  and $\pi_2=\nodelit{x:A}\rightlit{}\nodelit{z:B} \ \nodelit{\,}$. In $\pi_2$, the concatenated $\nodelit{\,}$ must match the same node as $z$ and thus it has no effect on the pattern; this accounts for the case where the node bound to $z$ has no incoming edge from a $C$-node. 

The pattern $\nodelit{x:A}\rightlit{y}\quantlit{\interval{1}{\infty}} \nodelit{z:B}$
looks for paths of arbitrary positive length from an $A$-node to a $B$-node.
It uses variable $y$ to bind edges encountered on this path. Unlike the bindings for $x$ and $z$, which are unique nodes, there may well be multiple edges encountered on the path between them, and thus $y$ needs to be bound to a complex object encoding all such edges on a path. Intuitively, in this case $y$ binds to the {\em list} of edges on a matching path. However, in general the binding for such variables, which we call {\em group variables}, is more complex. Indeed, consider a pattern $\pat\quantlit{\interval{n}{m}}$ and a particular match of this pattern in which $\pat$ is repeated $k$ times, $n\leq k \leq m$. These $k$ repetitions of $\pat$ are matched by paths $p_1,\ldots,p_k$, and thus for every variable used in $\pat$ we need to record not only which elements of $p_1\cdots p_k$ it binds to, but also in which paths $p_i$ these elements occur. Thus, in general, group variables will be bound to lists of (path, graph element)
pairs. 

The pattern $\big[\nodelit{x:A}\rightlit{y}\quantlit{\interval{1}{\infty}} \nodelit{z:B}\big]\condlit{x.a=z.a}$ is 
an example of a conditioned pattern; here the condition
ensures that the value of property $a$ is the same at the endpoints of the path. Note that conditions cannot compare nodes or edges, only their  properties. 

A pattern cannot be used by itself as a query; for example if we write $u=\big[\nodelit{x:A}\rightlit{y}\quantlit{\interval{1}{\infty}} \nodelit{z:B}\big]$ then the variable $u$ can be bound to infinitely many paths. Indeed, if there is a loop on some path from $x$ to $z$, it can be traversed arbitrarily many times, while still satisfying the condition of the pattern. To deal with this, every pattern in a query is compulsorily preceded by a restrictor, e.g., $u = \rtrail\ \big[\nodelit{x:A}\rightlit{y}\quantlit{\interval{1}{\infty}} \nodelit{z:B}\big]$. Then only trails, of which there are finitely many, that satisfy the conditions of $\pi$ will be returned as values of variable $u$.

\paragraph{The necessity of type rules}

The calculus defined in Fig.~\ref{f:patterns-fig} is very permissive and 
allows expressions that do not type-check. For example,
$\nodelit{x}\rightlit{x}\nodelit{\,}$ is syntactically permitted even though it
equates a node variable with an edge variable. As another example, adding the
condition $x.a=y.a$ to the pattern
$\nodelit{x:A}\rightlit{y}\quantlit{\interval{1}{\infty}} \nodelit{z:B}$ seen
above would result in comparing a singleton with a list of pairs.
The type system introduced next eliminates such mismatches.

\section{Type System}
\label{sec:types}

\begin{figure*}\centering

\newcommand{\custominference}[2]{$\displaystyle\inference{#1}{#2}$}

\custominference{}{\nodelit{x} \vdash x:\tnode}
\hfil
\custominference{}{\nodelit{x:\ell} \vdash x:\tnode}
\hfil
\custominference {}{\arclit{x}\  \vdash x:\tedge} 
\hfil
\custominference {}{\arclit{x :\ell}\  \vdash x:\tedge} 
\hfil
\custominference{ x \notin \var(\pat)}{x=\restrictor\ \pat \vdash x:\tpath}

\bigskip

\custominference {\pat\vdash z:\tau}{\pat\quantlit{n..m} \vdash z:\tlist(\tau)} 
\hfil
\custominference {\pat\vdash z:\tau}{\restrictor \pat \vdash z:\tau}
\hfil
\custominference{\pat \vdash z:\tau \quad z\neq x }{ x= \restrictor \pat \vdash  z:\tau}

\bigskip

\custominference {\pat\vdash x:\tau \quad \tau\in\set{\tnode,\tedge}}{\pat \vdash x.a=c:\tbool} \hfil
\custominference {\pat\vdash x:\tau\quad \pat\vdash y:\tau'\quad\tau,\tau'\in\set{\tnode,\tedge}}{\pat \vdash x.a=y.b:\tbool} 

\bigskip

\custominference {\pat\vdash \theta:\tbool \quad \pat\vdash\theta':\tbool}{\pat \vdash \theta\wedge\theta':\tbool} 
\hfil
\custominference {\pat\vdash \theta:\tbool \quad \pat\vdash\theta':\tbool}{\pat \vdash \theta\vee\theta':\tbool} 
\hfil
\custominference {\pat\vdash \theta:\tbool}{\pat \vdash \neg\theta:\tbool} 
\hfil
\custominference {\pat\vdash \theta:\tbool \quad \pat\vdash z:\tau}{\pat\condlit{\theta} \vdash z:\tau} 

\bigskip

\custominference {\pat_1\vdash z:\tau \quad \pat_2\vdash z:\tau}{\pat_1 + \pat_2 \vdash z:\tau}
\hfil
\custominference {\pat_1\vdash z:\tau \quad \pat_2\vdash z:\tmaybe(\tau)}{\pat_1 + \pat_2 \vdash z:\tmaybe(\tau)}
\hfil
\custominference {\pat_1\vdash z:\tmaybe(\tau) \quad \pat_2\vdash z:\tau}{\pat_1 + \pat_2 \vdash z:\tmaybe(\tau)}

\bigskip

\custominference {\pat_1\vdash z:\tau \quad  z \notin \var(\pat_2) }{\pat_1 + \pat_2 \vdash 
z: \tau?} 
\hfil
\custominference {\pat_2\vdash z:\tau \quad z \notin \var(\pat_1) }{\pat_1 + \pat_2 \vdash z: \tau?}

\bigskip

\custominference {\pat_1\vdash z:\tau \quad \pat_2\vdash z:\tau \quad \tau\in\set{\tnode,\tedge}}{\pat_1 \, \pat_2 \vdash z:\tau} 
\hfil
\custominference {\pat_1\vdash z:\tau \quad z \not\in\var(\pat_2)}{\pat_1 \, \pat_2 \vdash z:\tau}
\hfil
\custominference {\pat_2\vdash z:\tau \quad z \not\in\var(\pat_1)}{\pat_1 \, \pat_2 \vdash z:\tau} 

\bigskip

\custominference {Q_1\vdash z:\tau \quad Q_2\vdash z:\tau \quad \tau\in\set{\tnode,\tedge}}{Q_1 , Q_2 \vdash z:\tau}
\hfil
\custominference{Q_1\vdash z:\tau \quad z\not\in\var(Q_2) }{Q_1, Q_2 \vdash z:\tau}
\hfil
\custominference{Q_2\vdash z:\tau \quad z\not\in \var(Q_1) }{Q_1, Q_2 \vdash z:\tau}

\caption{Typing rules for the \gpml type system.}
\label{fig:type-system}
\end{figure*}

The goal of the type system is to ensure that \gpml expressions do not exhibit the pathological behavior explained at the end of the previous section. 

The set $\TT$ of types used to type variables is defined by the following grammar
\[ \tau\; \mbox{::=}\; \tnode \mid \tedge \mid \tpath \mid  \tmaybe(\tau) \mid \tlist(\tau).\]
The three atomic types are used for variables returning nodes, edges, and paths, respectively. The type constructor $\tmaybe$ is used for variables occurring on one side of a disjunction only, while $\tlist$ is used for variables occurring under repetition, whose bindings are grouped together.  As variables in \gpml are never bound to property values, we do not need the usual types like integers or strings. However, to eliminate references to unbound variables, we do need to type conditions (such as for $\langle x.a=y.a \rangle$ in the example at the end of Section~\ref{sec:calculus}); we use an additional type $\tbool$ for that.

Typing statements are of the form $\expr \vdash x:\tau$ stating that in expression $\expr$ (a pattern or a query), we can derive that variable $x$ has type $\tau$, and $\expr\vdash\theta:\tbool$, stating that a condition is correctly typed as a Boolean value under the typing of other variables.
 
The typing rules are presented in Figure~\ref{fig:type-system}. Here, $\var(\expr)$ stands for the set of variables used in expression $\expr$. 
For a type $\tau$ we let $\tau? = \tau$ if $\tau = \tmaybe(\tau')$ for some $\tau'$ and $\tau? = \tmaybe(\tau)$ otherwise.

The first five rules state that variables in node/edge patterns, and variables naming paths, are typed accordingly. The next four rules say that variables of group type appear in repetition patterns, and that restrictors and path naming do not affect typing. 

The next two lines deal with typing conditions: property values of singletons can be compared for equality; conditions are closed under Boolean connectives; and correctly typed conditions do not affect the typing of variables in a pattern.

The following two lines deal with the optional type $\tmaybe(\tau)$. It is assigned to a variable $z$ in a disjunction $\pat_1 + \pat_2$ if in one of the patterns $z$ is of type $\tau$ and in the other $z$ is either not present or of type $\tmaybe(\tau)$. 

Derivation rules for concatenation $\pat_1\pat_2$ and join $Q_1,Q_2$ are similar: a variable is allowed to appear in both expressions only if it is typed as a node or an edge in both, or it inherits its type from one when it does not appear in the other. 

\begin{definition}
An expression is \emph{well-typed}  if for every variable used in it, its type can be derived according to the typing rules.
\end{definition}

A well-typed expression  assigns a unique type to every variable appearing in it, and only to such variables. 

\begin{proposition}\label{l:type-sys is functional}
For every well-typed expression $\expr$, variable $x$, and types $\tau, \tau'$, 
\begin{itemize}
    \item $\expr\vdash x: \tau$ implies $x \in \var(\expr)$;
    \item $\expr \vdash x: \tau$ and $\expr\vdash x: \tau'$ imply that $\tau = \tau'$.
\end{itemize}
\end{proposition}
\begin{proof}
The first item holds because a variable appears in the conclusion of an inference rule only if it appears in one of its prem\-ises, or explicitly in the expression.
The second item holds because all inference rules have mutually exclusive premises.
\end{proof}

\begin{definition}
    Let~$\operator$ denote a binary operator from \gpml (Fig.~\ref{f:patterns-fig}).
    We say that~$\operator$ is \emph{associative  (resp.\@commutative) with respect to the type system}
    if the condition \eqref{e:associative} (resp.\@  \eqref{e:commutative}) below holds for all expressions~$\expr_1, \expr_2, \expr_3$, types~$\tau$, and  variables~$x$:
    \begin{gather}
    \label{e:associative}
        (\expr_1\operator\expr_2) \operator \expr_3 \vdash x:\tau 
            \iff  
        \expr_1\operator(\expr_2 \operator \expr_3)\vdash x:\tau\,,
    \\
        \expr_1\operator\expr_2\vdash x:\tau 
            \iff \expr_2\operator\expr_1\vdash x:\tau\,.
    \label{e:commutative}
    \end{gather}
\end{definition}

\begin{proposition} \mbox{}
    \begin{itemize}
        \item Union, concatenation and join are associative and commutative with respect to the type system.
        \item There is no expression $\expr$, variable~$x$, and type~$\tau$ such that $\expr\vdash x:\tmaybe(\tmaybe(\tau))$.
    \end{itemize}
    
\end{proposition}

\label{sec:schema}

A \emph{schema}~$\sigma$ is a partial function from variables $\XX$ to types $\TT$, with a finite domain. With each well-typed expression $\expr$ we can naturally associate a schema $\sch(\expr)$, induced by the types derived from $\expr$. It is defined formally below; it is well-defined by Proposition~\ref{l:type-sys is functional}.

\begin{definition}
Given a well-typed expression $\expr$, the \emph{schema of}~$\expr$, written~$\sch(\expr)$, is the schema that maps each variable~$x\in\var(\expr)$ to the unique type~$\tau$ such that~$\expr\vdash x:\tau$.
A variable~$x$ in~$\var(\expr)$ is called
\begin{itemize}
    \item a \emph{singleton} variable if~$\sch(\expr)(x)\in\set{\tnode,\tedge}$;
    \item a \emph{conditional} variable if~$\sch(\expr)(x)=\tmaybe(\tau)$ for some~$\tau$;
    \item a \emph{group} variable if~$\sch(\expr)(x)=\tlist(\tau)$ for some~$\tau$;
    \item a \emph{path} variable if~$\sch(\expr)(x)=\tpath$.
\end{itemize}
\end{definition}

\begin{remark}\label{remark:schema is compositional}
    {\em It is easily checked that the function~$\sch$ is compositional,
    in the sense that $\sch(\expr_1\operator\expr_2)$,  
    for some binary operator $\operator$ can be computed by a function that depends only on~$\operator$ and takes as arguments~$\sch(\expr_1)$ and~$\sch(\expr_2)$ (and likewise for unary operators)}.
\end{remark}

\section{Semantics}
\label{sec:semantics}
We begin by defining \emph{values}, which is \emph{what can be returned by a query}. Since \gpml returns references to graph elements, not  the data they bear, elements of $\const$ are not values.

\begin{definition}
Given a type $\tau\in\TT$, the set $\VVtype{\tau}$ of \emph{values of type~$\tau$} is defined inductively as follows
\begin{itemize}
    \item $\VVtype{\tnode}=\NN$\,,\; $\VVtype{\tedge}=\EE_\mathsf{d} \cup \EE_\mathsf{u}$\,,\;  $\VVtype{\tpath}=\PP$\,;
    \item $\VVtype{\tmaybe(\tau)}=\VVtype{\tau} \cup \set{\tnothing}$ for a special value $\tnothing$;
    \item $\VVtype{\tlist(\tau)}$ is the set of all composite values of the form \[\listval\big((p_1,v_1),\ldots,(p_n,v_n)\big)\] where
    $n \geq 0$ and $p_i\in\VVtype{\tpath}$ and $v_i\in\VVtype{\tau}$ for all $i \in [1,n]$.
\end{itemize}
The set of all values is $\VV = \bigcup_{\tau\in\TT} \VV_\tau$.
\end{definition}

The semantics of \gpml is defined in terms of \emph{assignments} binding variables to values. Formally, an assignment $\mu$ is a partial function from $\XX$ to $\VV$, with finite domain. We write $\square$ for the empty assignment; that is, an assignment that binds no variables. The values bound to variables of a well-typed pattern or query should respect its schema:
we say that an assignment~$\mu$ \emph{conforms} to a schema~$\sigma$ if~$\dom(\mu)=\dom(\sigma)$ and $\mu(x)\in\VVtype{\sigma(x)}$ for all~$x\in \dom(\mu)$. 

Given two assignments $\mu$ and $\mu'$, we say that $\mu$ and $\mu'$ \emph{unify} if $\mu(x) = \mu'(x)$ for all $x \in \dom(\mu) \cap \dom(\mu')$. In that case, we define their \emph{unification} $\mu\cup\mu'$ by setting 
$(\mu \cup \mu')(x) = \mu(x)$ if $x\in\dom(\mu)$ and $(\mu \cup \mu')(x) = \mu'(x)$ otherwise,
for all $x\in\dom(\mu)\cup\dom(\mu')$.
If $S$ is a family of assignments that pairwise unify, then their unification is associative, and we write it as $\bigcup S = \bigcup_{\mu \in S}\mu$.

The semantics of a well-typed \gpml expression $\expr$  on a property graph $G$ is a pair $\big(\sch(\expr),
\sem{\expr}_G\big)$, where $\sch(\expr)$ is the schema of $\expr$ (see
Section~\ref{sec:types}), and $\sem{\expr}_G$ is the \emph{set of answers} to
$\expr$ on $G$.

An \emph{answer} to $\expr$ on $G$ is a pair $(\bar p,\mu)$, where $\bar{p}$ is
a tuple of paths in $G$, and $\mu$ is an assignment that conforms to
$\sch(\expr)$. If $\expr$ is a pattern, $\bar{p}$ consists of a single path $p$,
in which case we simply write $p$ instead of $(p)$. If $\expr$ is a query,
$\bar{p}$ contains one path for each joined pattern.

By Remark~\ref{remark:schema is compositional}, $\sch(\expr)$ can be computed
compositionally independently of~$\sem{\expr}_G$. In what follows, we shall
define $\sem{\expr}_G$ using $\sch(\expr)$ and $\sem{\expr'}_G$ for direct
subexpressions $\expr'$ of $\expr$. Hence, the semantics of expressions
(patterns and queries), i.e., the function $\expr\mapsto(\sch(\expr),
\sem{\expr}_G)$, is compositional.

For the remainder of this section, we consider a fixed property graph $G =
\langle N, E_\mathsf{d}, E_\mathsf{u}, \lambda, \endpoints, \src, \tgt, \dv\rangle$.

\subsubsection*{Semantics of atomic patterns}

For the sake of brevity, here we write atomic patterns as if all components were
present, but still allow the possibility that some of them may be absent. Hence,
$\nodelit{x:\ell}$ subsumes the cases $\nodelit{x}$, $\nodelit{:\ell}$, and
$\nodelit{}$.
\begin{align*}
  \semoneline{\sem{\nodelit{x:\ell}}_G}{\setst*{(\pathval(n),\mu)}
    {n\in N,\, \ell\in\lambda(n) \text{ if $\ell$ is present}}}
  \\
  \intertext{where $\mu=\{x\mapsto n\}$ if $x$ is present, and~$\mu=\square$ otherwise.}
  \semoneline{\sem{\rightlit{x:\ell}}_G}{\setst*{(\pathval(u_1,e,u_2),\mu')}
    {\begin{array}{@{}l@{}}e\in E_d,\\
        u_1=\src(e),~u_2=\tgt(e),\\
        \ell\in\lambda(e) \text{ if $\ell$ is present}\end{array}}}
  \\
  \semoneline{\sem{\leftlit{x:\ell}}_G}{\setst*{(\pathval(u_2,e,u_1),\mu')}
    {\begin{array}{@{}l@{}}e\in E_d,\\
        u_1=\src(e),~u_2=\tgt(e),\\
        \ell\in\lambda(e) \text{ if $\ell$ is present}\end{array}}}
  \\
  \semoneline{\sem{\undirlit{x:\ell}}_G}{\setst*{(\pathval(u_1,e,u_2),\mu')}
    {\begin{array}{@{}l@{}}e\in E_u,
        \\ 
        \endpoints(e)= \{u_1,u_2\},\\
        \ell\in\lambda(e) \text{ if $\ell$ is present}\end{array}}}
\end{align*}
where $\mu'=\{x\mapsto e\}$ if $x$ is present, and $\mu'=\square$ otherwise.
Observe that the pattern $\undirlit{x:\ell}$ returns both $\pathval(u_1,e,u_2)$
and $\pathval(u_2,e,u_1)$ if $\endpoints(e) = \{u_1,u_2\}$ with $u_1 \neq u_2$,
but only one of them if $u_1 = u_2$ since both paths are the same.

\subsubsection*{Semantics of concatenation}

\begin{align*}
  \semoneline{\sem{\pi_1 \, \pi_2}_G}{%
    \setst*{(p_1 \cdot p_2, \mu_1 \cup \mu_2)}{%
      \begin{array}{@{}l@{}}
        (p_i,\mu_i) \in \sem{\pi_i}_G \text{ for } i=1,2,\\
        p_1 \text{ and } p_2 \text{ concatenate}, \\
        \mu_1 \text{ and } \mu_2 \text{ unify} 
      \end{array}%
    }%
  }%
\end{align*}
The typing system ensures that all variables shared by $\pat_1$ and $\pat_2$ are
singleton variables (otherwise $\pat_1\pat_2$ would not be well-typed). In other
words, implicit joins over group and optional variables are disallowed (path
variables do not occur in patterns at all).

\subsubsection*{Semantics of union}

\begin{equation*}
  \sem{\pi_1 + \pi_2}_G =
  \setst*{(p, \mu \cup \mu')}{(p,\mu)\in\sem{\pat_1}_G\cup\sem{\pat_2}_G}
\end{equation*}
where $\mu'$ maps every variable in $\dom\big(\sch(\pi_1 + \pi_2)\big) \setminus
\dom(\mu)$ to $\tnothing$. Note that here we rely on $\sch(\pi_1 + \pi_2)$.

\subsubsection*{Semantics of conditioned patterns}

\begin{equation*}
  \sem{\pi\condlit{\theta}}_G = \setst*{ (p, \mu) \in
    \sem{\pi}_G}{\mu\models\theta}
\end{equation*}
where $\mu\models\theta$ is defined inductively as follows:
\begin{itemize}
\item $\mu\models (x.a =c)$ iff $\dv\bigl(\mu(x),a\bigr)$ is defined and
  equal to $c$;
\item $\mu\models (x.a=y.b)$ iff $\dv(\mu(x),a)$ and $\dv\bigl(\mu(y),b\big)$
  are defined and equal;
\item $\mu\models (\theta_1 \wedge \theta_2)$ iff $\mu\models\theta_1$ and $\mu\models\theta_2$;
\item $\mu\models (\theta_1 \vee \theta_2)$ iff $\mu\models\theta_1$ or $\mu\models\theta_2$;
\item $\mu\models(\neg \theta)$ iff $\mu\not\models\theta$.
\end{itemize}

\subsubsection*{Semantics of repeated patterns}

\begin{equation*}
\sem{\pi\quantlit{n..m}}_G = 
\bigcup_{\mathclap{i=n}}^{m} \sem{\pi}_G^i
\end{equation*}
Above, for a pattern $\pat$ and a natural number $n\geq 0$, we use
$\sem{\pat}_G^n$ to denote the $n$-th power of $\sem{\pat}_G$, defined as
follows. We let
\begin{align*}
  \semoneline{\sem{\pat}_G^0}{\setst{(\pathval(u),\mu)}{u\text{ is a node in }G}}
\end{align*}
where $\mu$ is the assignment that maps each variable in $\dom(\sch(\pat))$
to~$\listval()$, the empty composite value. For $n>0$, we let
$$\sem{\pat}_G^n\ = \ 
    \setst*{(p,\mu)}{\begin{array}{@{}l@{}}
        (p_1,\mu_1), \ldots, (p_n,\mu_n) \in \sem{\pat}_G\\
        p=p_1\cdot \ldots \cdot p_n\\
      \mu= \collect
      \big((p_1,\mu_1),\ldots,(p_n,\mu_n)\big) 
      \end{array}}
$$
where $\collect[(p_1,\mu_1),\ldots,(p_n,\mu_n)]$ is an assignment defined, and
discussed, below. In the case that $\pat$ does not have any variables, $\collect$ simply returns a function with empty domain. If $\pi$ does contain variables, then each such variable is mapped to a \emph{list}. (As such, nesting of patterns of the form $\pat^{n..m}$ leads to nesting of 
lists.)
 
There are several ways to define $\collect$ and obtain a sound semantics.
In all cases, $\collect$ takes as input any number of path/binding pairs $(p_1,\mu_1),\ldots,(p_n,\mu_n)$, such that~$n > 0$ and $p_1, \ldots,p_n$ concatenate to a path~$p=p_1\cdots p_n$. Furthermore, by our inductive definition of the semantics, it will always be the case that 
$\mu_1,\ldots,\mu_n$ all have the same domain, that we denote by~$D$.
Then, $\collect\big((p_1,\mu_1),\allowbreak \ldots,(p_n,\mu_n)\big)$ is an assignment that maps every~$x\in D$ to a $\listval((p_1',\allowbreak v_1),\ldots,(p_\ell',v_\ell))$ where each~$p'_i$ is a portion of the matched path (they collectively satisfy $p_1'\cdots p_\ell'=p_1\cdots p_n$) and~$v_i$ is the value associated to~$x$ for that portion.

If all $p_i$'s have a positive length (i.e., have at least one edge), $\collect$ is simply defined as follows.
\begin{multline}
\label{naive-collect-eq}
    \forall x\in D\quad
    \collect\big((p_1,\mu_1),\ldots,(p_n,\mu_n)\big)(x)= \\ \listval((p_1,\mu_1(x)),\ldots,(p_n,\mu_n(x)))
\end{multline}
Although $\collect$ is still well defined by \eqref{naive-collect-eq} if some of the $p_i$'s have length~$0$, the above definition may lead to infinite query results. To avoid this, we outline three different approaches. 

\paragraph{Approach 1: Syntactic restrictions} 
We add a syntactic restriction that prevents the case from ever appearing: pattern $\pat\quantlit{n..m}$ is forbidden if pattern~$\pat$ may match an edgeless path. The latter is defined inductively: every edge pattern is allowed; if $\pat$ is allowed then so are $\pat\condlit{\theta}$, $\pat\quantlit{\interval{n}{m}}$, $\pat\pat'$ and $\pat'\pat$ for every condition $\theta$, $n>0$, and pattern $\pat'$; if $\pat_1$ and $\pat_2$ are allowed then so is $\pat_1+\pat_2$. 
This is the solution adopted by the GQL standard : the \emph{minimum path length} of~$\pat$ must be positive. The drawback of this solution is that it rules out syntactically some patterns for which (\ref{naive-collect-eq}) would result in a well-defined finite semantics.

\paragraph{Approach 2: Run-time restriction} As an alternative, the precondition for $\collect$ well-definedness can be checked at run-time, i.e., it is only defined if all $p_i$'s have a positive length.
While not imposing any additional restrictions, this approach has a drawback that $\pat$ may have some result while $\pat\quantlit{1..1}$ has none, for some pattern $\pat$.

\begin{figure}[t]
\begin{tikzpicture}[yscale=.4,font=\small]
\newcommand{\absi}{0}
\newcommand{\absii}{0.75}
\newcommand{\absiii}{2.5}
\newcommand{\absiv}{5}
\newcommand{\absv}{6.5}

\newcommand{\horisep}{.05}

\newcommand{\rectii}{0.5}
\newcommand{\rectiii}{1.5}
\newcommand{\rectiv}{2}

\draw [draw=black,line width=1pt] 
    ($(\absi,\rectii)+(\horisep,0)$) 
    rectangle 
    ($(\absii,\rectiii)-(\horisep,0)$) 
    node[pos=0.5] (p1) {$p_1$}
    ;
\draw [draw=black,line width=1pt] 
    ($(\absii,\rectii)+(\horisep,0)$) 
    rectangle 
    ($(\absiii,\rectiii)-(\horisep,0)$) 
    node[pos=0.5] {$p_2$};
\draw [draw=black,line width=1pt] 
    ($(\absiii,\rectii)+(\horisep,0)$) 
    rectangle 
    ($(\absiv,\rectiii)-(\horisep,0)$) 
    node[pos=0.5] {$p_6$};
\draw [draw=black,line width=1pt] 
    ($(\absiv,\rectii)+(\horisep,0)$) 
    rectangle 
    ($(\absv,\rectiii)-(\horisep,0)$) 
    node[pos=0.5] {$p_8$};

\draw [decorate, decoration = {calligraphic brace}] 
    ($(\absiii,\rectiv)+(0.8,0)$,\rectiv)  
    --      
    ($(\absiii,\rectiv)-(0.8,0)$)
    node[anchor=base,pos=0.5,yshift=3pt] {$p_3=p_4=p_5$}
    ;
    
\draw [decorate, decoration = {calligraphic brace}] 
    ($(\absiv,\rectiv)+(0.2,0)$) 
    --
    ($(\absiv,\rectiv)-(0.2,0)$)   
    node[anchor=base,pos=0.5,yshift=3pt] {$p_7$}
    ;

\draw [decorate, decoration = {calligraphic brace}] 
    ($(\absv,\rectiv)+(0.5,0)$)
    -- 
    ($(\absv,\rectiv)-(0.5,0)$) 
    node[anchor=base,pos=0.5,yshift=3pt] {$p_9=p_{10}$};
\path let \p1=(p1.base) in (\absi,\y1) node[anchor=base east]{$p={}$};

\draw [draw=black,line width=1pt]    
    ($(\absi,-\rectii)+(\horisep,0)$) 
    rectangle 
    ($(\absii,-\rectiii)-(\horisep,0)$)  
    node[pos=0.5] (p1p) {$p'_1$};
\draw [draw=black,line width=1pt]   
    ($(\absii,-\rectii)+(\horisep,0)$) 
    rectangle 
    ($(\absiii,-\rectiii)-(\horisep,0)$)  
    node[pos=0.5] {$p'_2$};
\draw [draw=black,line width=1pt] 
    ($(\absiii,-\rectii)+(\horisep,0)$) 
    rectangle 
    ($(\absiv,-\rectiii)-(\horisep,0)$)  
    node[pos=0.5] {$p'_4$};
\draw [draw=black,line width=1pt]     
    ($(\absiv,-\rectii)+(\horisep,0)$) 
    rectangle 
    ($(\absv,-\rectiii)-(\horisep,0)$)   
    node[pos=0.5] {$p'_6$};

\draw [decorate, decoration = {calligraphic brace}]    
    ($(\absiii,-\rectiv)-(0.2,0)$)   
    -- 
    ($(\absiii,-\rectiv)+(0.2,0)$,\rectiv) 
    node[anchor=base,pos=0.5,yshift=-7pt] {$p'_3$}
    ;
\draw [decorate, decoration = {calligraphic brace}]
    ($(\absiv,-\rectiv)-(0.2,0)$) 
    --
    ($(\absiv,-\rectiv)+(0.2,0)$)  
    node[anchor=base,pos=0.5,yshift=-7pt] {$p'_5$}
    ;
\draw [decorate, decoration = {calligraphic brace}] 
    ($(\absv,-\rectiv)-(0.2,0)$)
    --
    ($(\absv,-\rectiv)+(0.2,0)$) 
    node[anchor=base,pos=0.5,yshift=-7pt] (pp7) {$p'_7$}
    ;
\path let \p1=(p1p.base) in (\absi,\y1) node[anchor=base east]{$p={}$};

\path
    let \p1=(current bounding box.east),
        \p2=(current bounding box.north) 
    in (\x1,\y2) coordinate (bbne);
\path
    let \p3=(current bounding box.west),
        \p4=(pp7.base)
    in (\x3,\y4) coordinate (bbsw);
\pgfresetboundingbox;
\path[use as bounding box] (bbne) rectangle (bbsw);
\end{tikzpicture}%
\caption{Refactorization of a path $p=p_1p_2\cdots p_{10}$ as $p=p'_1p'_2\cdots p'_{7}$ by grouping consecutive edgeless factors}
\label{fig:refactoring}
\end{figure}

\paragraph{Approach 3: Grouping edgeless paths} To overcome problems with the first two approaches, we propose a more general semantics of 
$\collect$ that groups together consecutive edgeless paths from $p_1,\cdots,p_n$. If no such paths exist, either due to syntactic restriction or ruling them out at run-time, the result of this approach coincides with (\ref{naive-collect-eq}); thus this approach subsumes the other two.

We define~$p'_1,\cdots,p'_{\ell}$ as a coarser factorization of~$p_1\cdots p_n$:
each $p'_i$ is the concatenation of successive $(p_j)$'s, in which consecutive edgeless paths are grouped together,
as shown in Figure~\ref{fig:refactoring}.
Formally, the $p_i'$s are defined as the unique path sequence such that there exists~$i_1<i_2<\cdots<i_{\ell+1}$ (delineating the boundaries of $p_1',\cdots,p_\ell'$) with~$i_1=1$, $i_{\ell+1}=n+1$ and satisfying the following.
\begin{align*}
    \forall k\in \set{1,\ldots,\ell}\quad &p_k' = p_{i_k}\cdots p_{i_{k+1}-1}
    \\
    \forall k\in\set{1,\ldots,\ell-1}\quad & \pathlen(p_{i_k}) \neq 0 \mathrel{\vee} \pathlen(p_{i_{k+1}}) \neq 0 
    \\
    \forall k\in \set{1,\ldots,\ell}\quad & \begin{cases}
    \text{either} & i_{k+1} = i_{k}+1 \text{ and } \pathlen(p_{i_k})\neq0\\
    \text{or}     & 
    \forall i~,i_k\leq i < i_{k+1},~\pathlen(p_{i})=0
    \end{cases}
    \end{align*}
The assignment $\collect\big((p_1,\mu_1),\ldots,(p_n,\mu_n)\big)$ is defined only if 
\begin{equation*}
    \forall k\in\set{1,\ldots,\ell}\quad \mu_{i_k},\ldots,\mu_{(i_{k+1}-1)}\text{ pairwise unify}
\end{equation*}
Then their unification is denoted by~$\mu'_k$ 
and $\collect$ is defined by 
\begin{multline*}
    \forall x\in D\quad
    \collect\big((p_1,\mu_1),\ldots,(p_n,\mu_n))(x) =\\
    \listval((p'_1,\mu'_1(x)),\cdots,(p'_\ell,\mu'_\ell(x))\big)
\end{multline*}

\begin{remark}
    For the purpose of $\collect$, one could use a weaker definition
    for unification that would allow~$\mu$ and~$\mu'$ to unify if,
    for every~$x\in\dom(\mu)\cap\dom(\mu')$, any of the following holds:
    $\mu(x)=\tnothing$, $\mu'(x)=\tnothing$ or $\mu(x)=\mu'(x)$.
    This would allow even more combinations than the definition above.
\end{remark}

\subsubsection*{Semantics of queries} 

\begin{align*}
  \semoneline{\sem{\rtrail\ \pi}_G}{%
    \setst*{ (p, \mu) \in \sem{\pi}_G}{%
      \begin{array}{@{}l@{}}
        \text{no edge occurs more}\\
        \text{than once in } p
      \end{array}%
    }%
  }
  \\
  \semoneline{\sem{\rsimple\ \pi}_G}{%
    \setst*{ (p, \mu) \in \sem{\pi}_G}{%
      \begin{array}{@{}l@{}}
        \text{no node occurs more}\\
        \text{than once in } p
      \end{array}%
    }%
  }
  \\
  \semtwoline{\sem{\shortest~\expr}_G}{
    \setst*{(p,\mu)\in \sem{\expr}_G}{%
      \begin{array}{@{}l@{}}
        \pathlen(p) = \min\setst*{\pathlen(p')}{
          \begin{array}{@{}l@{}}
            (p',\mu')\in \sem{\expr}_G\\
            \src(p')=\src(p)\\
            \tgt(p')=\tgt(p)\\
          \end{array}%
        }%
      \end{array}%
    }%
  }
\end{align*}
where $\expr$ is $\pat$, $\rtrail\ \pat$ or $\rsimple\ \pat$, for some
pattern $\pat$.
We then define:
\begin{align*}
  \sem{x = \restrictor{}\pat}_G & = 
    \big\{ (p, \mu \cup \{x \mapsto p\}) \mid
     (p,\mu) \in \sem{\rho\ \pat}_G \big\}%
     \\
  \sem{Q_1, Q_2}_G & =  %
    \left\{\,(\bar{p}_1 \times \bar{p}_2,\mu_1 \cup \mu_2)
    ~\middle|~
      \begin{array}{@{}l@{}}
        (\bar{p}_i,\mu_i) \in \sem{Q_i}_G \text{ for } i=1,2\\
        \mu_1 \text{ and } \mu_2 \text{ unify}
      \end{array}%
    \,\right\}
\end{align*}

\noindent Here, $\bar p_1 = (p_1^1, p_1^2, \dots, p_1^k)$ and $\bar p_2 = (p_2^1, p_2^2, \dots, p_2^l)$ are tuples of paths, and $\bar{p}_1 \times \bar{p}_2$ stands for $(p_1^1, p_1^2, \dots, p_1^k, p_2^1, p_2^2, \dots, p_2^l)$. Note that $\bar p_i$ is a single path when $Q_i$ does not contain the join operator.
Moreover, like for concatenation, the typing system guarantees that
$Q_1$ and $Q_2$ are only joined over singleton variables, but not over path,
group, or conditional variables.

In the results below we assume the third approach to the definition of
$\collect$ as subsuming the other two.
One may verify, by routine inspection, that the semantics is consistent with the typing system.

\begin{proposition}
\label{well-typed-prop}
For every well-typed expression~$\expr$ and every $(\bar{p},\mu)$ in $\sem{\expr}_G$, all paths in~$\bar{p}$ belong to~$\PP(G)$
and~$\mu$ conforms to $\sch(\expr)$.
\end{proposition}

Even though the set $\Paths(G)$ may be infinite, syntactic restrictions ensure finiteness of output.

\begin{theorem} \label{t:hereditarily finite}
$\sem{Q}_G$ is finite for each query $Q$ and graph $G$.
\end{theorem}

\begin{proof}[Proof Sketch]
We will treat only the case when $Q$ is $\rho \pat$; other cases follow from this case or are straightforward.

Let us first show that the set~$P=\setst*{p}{\exists \mu~ (p,\mu)\in\sem{Q}_G}$ is finite.
If~$\rho$ is one of $\rtrail$, $\rsimple$, $\shortest\ \rtrail$ or $\shortest\ \rsimple$,
the claim holds since there are finitely many trails and simple paths in a graph.
The last case, that is~$\rho=\shortest$, follows from the fact that for all nodes $s$ and $t$ the set
\begin{equation*}
    P_{(s,t)}=\setst{p\in P}{p\text{ starts in }s\text{ and ends in }t}
\end{equation*}
is finite. Indeed, all paths in~$P_{(s,t)}$ have the same length, and there are finitely many paths of a given length in a graph.

The remainder of the proof of Theorem~\ref{t:hereditarily finite} amounts to showing that for each path~$p\in P$, there are finitely many~$\mu$ such that~$(p,\mu)\in\sem{Q}_G$.
The proof of that claim is done by induction, the only nontrivial case is for patterns of the shape~$\pat\quantlit{n..\infty}$.
For details, see Lemma~\ref{l:finitely many assignments}, page~\pageref{l:finitely many assignments}, in the Appendix.
\end{proof}

\section{Expressivity and Complexity}
\label{sec:complexity}

\paragraph{Expressive power.} First, we look at the expressive power of \gpml.
For this, we will compare GPC with main graph query languages considered in the research literature. Specifically, we compare \gpml with regular path queries (RPQs)~\cite{RPQ,MendelzonW95}, and their two-way extension, 2RPQs~\cite{C2RPQ}. In essence, an RPQ is specified via a regular expression and returns all pairs of nodes connected by a path whose edge labels form a word in the language of this expression. 2RPQs also allow traversing edges in the reverse direction, similarly to $\leftlit{a}$ in \gpml, for a label $a$. Two natural extensions are C2RPQs, which close 2RPQs under conjunctions~\cite{CRPQ,C2RPQ}, and their unions, called UC2RPQs~\cite{CRPQ}. An interesting class is also that of nested regular expressions (NREs)~\cite{NRE}, where along a path conforming to a regular language, we can test if there is an outgoing path conforming to another regular expression, as in PDL, or XPath. Finally, we consider the class of regular queries (RQs)~\cite{RRV17}, which subsumes all the aforementioned classes. A regular query is  a non-recursive Datalog program that is allowed to use \emph{transitive atoms} of the form $R^{+}(x,y)$ in the body of the rules, where $R$ is a binary predicate, either built in, or defined in the program. 

In order to compare with the aforementioned languages, we  consider a simple extension of \gpml with projection and union, reflecting the fact that the pattern matching mechanism we formalize will be a sublanguage of a fully-fledged query language like GQL or SQL/PGQ. A \emph{$\gpmlplus$ query} is a set of rules 
\[ \ans(\bar x)\colondash Q_1; \quad \ans(\bar x) \colondash Q_2; \quad \dots \quad \ans(\bar x) \colondash Q_k \]
where $Q_i$ is a \gpml query such that $\bar x \subseteq \mathrm{var}(Q_i)$ for all $i$.
The semantics of such a query on graph $G$ is  \[\sem{Q_1}_G^{\bar x} \cup \sem{Q_2}_G^{\bar x} \cup \dots \cup \sem{Q_k}_G^{\bar x}\]
where $\sem{Q_i}_G^{\bar x} = \left\{\, \mu(\bar{x}) \mid \exists p\, (p,\mu) \in \sem{Q_i}_G\right\}$ for all $i$. Notice that in our definition we allow unions only at the top level in order to combine results of queries whose arity is higher than binary. Binary unions are already covered at the level of \gpml patterns, and can be arbitrarily nested inside iterations.

\begin{theorem}
\label{prop:expr}
\gpmlplus  can express all of the following:
\begin{itemize}
\item unions of conjunctive two-way regular path queries (UC2RPQs);
\item nested regular expressions (NREs);
\item regular queries.
\end{itemize}
\end{theorem}
\begin{proof}[Proof sketch]
For 2RPQs, note that these are explicitly present in the syntax of \gpml patterns, and projecting on the endpoints gives us an equivalent expression. C2RPQs and their unions are then handled by the conjunction of \gpml queries, and unions in \gpmlplus, respectively. The case of NREs is a bit more interesting, and it contains the blueprint for regular queries. To illustrate the main ideas, consider the nested regular expression $(a[b\quantlit{+}]c)\quantlit{+}$, which looks for paths of the form $(ac)^n$, where after traversing an $a$, we also check the existence of a nonempty path labelled with $b$s.
An equivalent \gpmlplus query is $\ans(x,y) \colondash Q$ where $Q$ is the \gpml query
$$\shortest~ (x) [ \rightlit{:a} (z) \rightlit{:b}\quantlit{\interval{1}{\infty}} () \leftlit{}\quantlit{\interval{1}{\infty}} (z) \rightlit{:c}]\quantlit{\interval{1}{\infty}} (y)\,.$$
Basically, we introduce a fresh variable $z$ which binds to the node from which we need to find a nonempty $b$-labelled path to an anonymous node, and then we return to the same node, thus allowing us to encode an answer inside of a single path. Since we only care about the endpoints in all of these query classes, the restrictor $\shortest$ is enough. This idea is then applied inductively in order to capture regular queries.
\end{proof}

\paragraph{Complexity.}
When it comes to evaluating graph queries, one is used to dealing with high complexity. For instance, checking whether there is a query answer with a restrictor $\rsimple$ or $\rtrail$ on top is known to be $\np$-hard in data complexity~\cite{MendelzonW95,B13,Woo,BaganBG20,MartensNT20}, and yet this feature is supported both by the GQL standard~\cite{sigmod22} and by concrete languages such as Cypher~\cite{Cypher}. Accepting such high complexity bounds probably stems from the fact that query answers can be large in case of graph queries. 
Indeed, there is no need to use convoluted 3-SAT reductions to have the engine run forever;
one may just enumerate all simple paths in the graph with query: $\simplelit \rightlit{}\quantlit{\interval{0}{\infty}}$~.

In this light, we provide some insights on computing answers of \gpml queries, i.e., we study the following enumeration problem:

\medskip

\fbox{
    \vspace{-2mm}
  \begin{tabular}{lp{40ex}}
    \textsc{Problem}: & \textsc{Enumerate answers}\\
    \textsc{Input}: & A property graph $G$, and a query $Q$.\\
    \textsc{Output}: & Enumerate all pairs $(p,\mu)\in \sem{Q}_G$ without repetitions.
  \end{tabular}
    \vspace{-2mm}}

\medskip

A potential criticism we would like to address is the fact that the path $p$, witnessing the output mapping $\mu$, is also returned each time, which can make query answers larger than necessary. The reason that we study the problem like this, however, is that this is what the GQL standard asks for. 
In our analysis, we will use Turing machines with output tape (in order to enumerate the results), and will bound the size of work tape the machine uses. The main result of this section is the following:

\begin{theorem}
\label{th:compl}
The problem \textsc{Enumerate answers} can be solved by a Turing machine using exponential space (in $G$ and $Q$). If we consider the query $Q$ to be fixed (data complexity), then the machine uses only polynomial space.
\end{theorem}
\begin{proof}[Proof sketch]
The basic idea is to enumerate all possible answers $(p,\mu)$ in increasing length of $p$, and check, one by one, whether they should be output. If we consider a single pattern with a restrictor on top, e.g., $Q = \rho \pi$, this approach works as described, and the size of the possible paths (and thus also mappings $\mu$), can be bounded by a size that is exponential in the size of $Q$ and $G$, and polynomial if we assume $Q$ to be fixed. For each such answer, we can validate whether it should be output in polynomial space. Notice that once a result is output, we can discard it, and move to the next one. Enumeration stops once an appropriate path length has been reached, and the next mapping $\mu$ is considered. Joins can then be evaluated by nesting this procedure.
\end{proof}

\begin{theorem}\label{theo:complexity-lower}
  The problem \textsc{Enumerate answers} cannot be solved by a Turing machine using polynomial amount of space (in $G$ and $Q$). 
\end{theorem}

\section{Looking ahead}
\label{sec:extensions}

In this section, we discuss possible extensions of \gpml\ that would reflect
additional features envisioned in GQL and SQL/PGQ. In doing so, we also provide
two examples of how theoretical research has directly influenced the drafts of
the GQL and SQL/PGQ standards as they were being written.

\paragraph{Placement of restrictors}

We imposed strict requirements for placing restrictors: optional $\shortest$
followed by optional $\traillit$ or $\simplelit$, with at least one of the three
present to ensure that the number of returned paths is finite. %
It is natural to wonder whether restrictors could be mixed arbitrarily, by
allowing patterns $\restrictor\,\pat$ where $\restrictor$ is one of $\shortest$,
$\traillit$, and $\simplelit$. In fact, this was an earlier proposal in the GQL
and PGQ drafts, which was then significantly modified.
To see why, consider the following graph

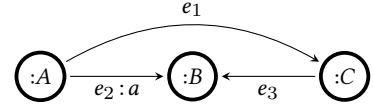
\begin{center} 
  \begin{tikzpicture}[node distance=2cm]
    \node[nodeid](l1){:$A$};
    \node[nodeid,right of=l1](l2){:$B$};
    \node[nodeid,right of=l2](l3){:$C$};
    
    \path[draw,->](l1) edge[right] node[below]{$e_2\!:\!a$} (l2);
    \path[draw,->](l3) edge[right] node[below]{$e_3$} (l2);
    \path[draw,->,bend left](l1) edge[bend left] node[above]{$e_1$}(l3);
  \end{tikzpicture}
\end{center}
(where $e_1, e_2$, and $e_3$ are edge ids) and the pattern
\begin{equation*}
  \traillit\Big[ \big[\shortest\ \nodelit{:A} %
  \rightlit{x}\quantlit{\interval{0}{\infty}} \nodelit{:B}\big] %
  \nodelit{:B} \leftlit{y:a}\quantlit{*} \nodelit{:A} \Big]\;.
\end{equation*}
Matching the subpattern outside $\shortest$ produces the
assignment of $y$ to the edge $e_2$. In the GQL rationale, $\shortest$ should restrict query answers in the sense that, out of all the answers to the query, it chooses the one with the shortest witness. If we follow this rationale, then, 
to keep the entire match a trail,
the group variable $x$ must be assigned the list $[e_1,e_3]$. Therefore,
counter-intuitively, a shortest match occurring under the scope of $\traillit$
produces a path that is not shortest between two nodes.

As a result, GQL pattern matching now disallows arbitrary mixing of
restrictors. At the same time, it is slightly more permissive than the version
presented here: $\shortest$ must appear at the top of a pattern, but $\traillit$
and $\simplelit$ can be mixed freely. Adding this feature is a possible
extension of \gpml.

\paragraph{Aggregation}

As one navigates along a path in a graph, aggregation is a natural feature for
computing derived quantities, such as path length. For instance, with
$\nodelit{:A} \rightlit{x}\quantlit{\interval{0}{\infty}}\nodelit{:B}$ looking for paths between
$A$ and $B$, one could return the total length $\sum x.\textit{length}$ of each
matched path.
However, adding aggregation is problematic. To see why, consider the simple
aggregate $\cnt x$ for a group variable $x$, which counts the number of
bindings of that variable. Now assume we extend the language with {\em
  arithmetic conditions} of the form $t_1 = t_2$, where $t_1$ and $t_2$ are
terms built from values $y.k$ and $\cnt x$ by means of addition ``$+$'' and
multiplication ``$\cdot$''. 
These already pack huge expressive power:

\begin{proposition}
  \label{diophantine-prop}
  The data complexity of \gpml\ with arithmetic conditions is undecidable.
\end{proposition}

\noindent In view of this, the current approach of GQL is to only allow
aggregates in the outputs of queries (no operations on them are permitted). But
it is a general open direction to understand how to tame and use the power of
aggregates in a graph language.

\paragraph{Bag semantics}

Relational calculus (first-order logic) is interpreted under set semantics, but
SQL uses bag semantics, and so does GQL. In the basic version of \gpml\ we opted
for set semantics, following the relational calculus precedent, but it is
necessary to study bag semantics as an extension of \gpml.

\paragraph{Nulls and bound conditional variables}

We have left out the treatment of nulls, assuming that conditions involving
non-applicable nulls --- i.e., values $\dv(x,k)$ where the property $k$ is not
defined for $x$ --- evaluate to {\em false}. Following SQL (and the current GQL
proposals) they would instead evaluate to {\em unknown}, leading to many known
issues \cite{CGLT20}. In addition, one could expand the language with a
predicate that checks whether a conditional variable is bound, as done, in fact,
in GQL.

\paragraph{Scoping of variables}

Consider the pattern $\nodelit{x} \rightlit{} \big[\nodelit{y} \rightlit{}
\nodelit{z}\big] \condlit{\theta}$, where $\theta$ is $x.k=y.k+z.k$; note, in
particular, the non-local occurrence of $x.k$ in the condition. Should this be
allowed? At first it seems innocent (we bind $x$ before evaluating the rest of
the pattern), but in a pattern like $\nodelit{x} \rightlit{} \big[ \nodelit{y}
\rightlit{} \nodelit{z} \big] \condlit{x.k+u.k=y.k+z.k} \rightlit{}
\nodelit{u}$ we need to evaluate the edge from $y$ to $z$ first. Moreover, if we
have repetitions, the evaluation procedure becomes much less clear. Despite
this, GQL plans to offer such kind of features.

\paragraph{Label expressions}

GQL will offer complex label expressions~\cite{sigmod22}, and these too can be added to the
calculus as an extension.

\paragraph{Enhancing conditions}

We assumed that data values come from one countable infinite set of constants,
very much in line with the standard presentations of first-order logic. In
reality, of course, data values are typed, and such typing must be taken into
account (at the very least, for the study of aggregation, to distinguish
numerical properties). Furthermore, one could permit explicit equalities between
singleton variables in conditions (currently, such variables implicitly join
when they are repeated in a pattern).
\begin{acks}
    The authors would like to thank the anonymous referees of PODS for
    their valuable comments and helpful suggestions. 
    This work is supported by:

    French's \grantsponsor{ANR}{Agence Nationale de la Recherche}{https://anr.fr} 
    under grant \grantnum{ANR}{ANR-18-CE40-0031};
    Poland's \grantsponsor{NCN}{National Science Centre}{https://www.ncn.gov.pl/} 
    under grant \grantnum{NCN}{2018/30/E/ST6/00042};
    ANID Fondecyt Regular project number 1221799.
\end{acks}

\bibliographystyle{ACM-Reference-Format}
\bibliography{references}

\appendix
\onecolumn 

\section{Semantics}
\label{sec:app-semantics}

The following lemma is needed for the proof of Theorem~\ref{t:hereditarily finite}.

\begin{lemma}\label{l:finitely many assignments}
    Let~$G$ be a graph,~$\pat$ be a well-typed pattern. For every path~$p$ of~$G$, there are finitely many~$\mu$ such that~$(p,\mu)\in\sem{\pat}_G$.
\end{lemma}
\begin{proof} 
By induction over the structure of~$\pat$.  The lemma is obviously true for all base cases, 
and all inductive cases are easy except for a repeated pattern of the form~$\pat\quantlit{h..\infty}$, for some~$h$.
\begin{equation*}
    \sem{\pat\quantlit{h..\infty}}_G = \bigcup_{i=h}^{\infty} \sem{\pat}_G^i ~.
\end{equation*}
The proof amounts to showing that there exists a bound~$B$ such that
\begin{equation}\label{eq:lemma-reduce-n}
    \forall n\mathbin{>}B,~
    \forall\mu, (p,\mu)\mathbin{\in}\sem{\pat}_G^n\quad 
    (p,\mu)\in\sem{\pat}_G^{n-1}~.
\end{equation}
We fix~$B=(L+1)(M+1)$ where
\begin{align*}
    L ={}& \pathlen(p) \\
    M ={}& \max \setst*{ 
                \card\setst{\mu}{(\pathval(u),\mu)\in\sem{\pat}_G} 
              }{
                u\text{ is a node in }p
            }%
\end{align*}
Note that $M$ is well defined since by IH, $\setst{\mu}{(\pathval(u),\mu)\in\sem{\pat}_G}$ is finite for every node~$u$.

Let~$n$ and~$\mu$ be as in \eqref{eq:lemma-reduce-n}. 
Since~$(p,\mu)\mathbin{\in}\sem{\pat}_G^n$, there are paths~$p_1,\ldots, \allowbreak p_n$
and assignments~$\mu_1,\ldots,\mu_n$ such that~$p=p_1p_2\cdots p_n$ and~$\mu=\collect((p_1,\mu_1),\ldots,(p_k,\mu_n))$.

We let~$\ell$ and the sequences~$(i_k)_{1\leq k < \ell}$ and~$(\mu'_k)_{1\leq k < \ell}$ be defined as in Section~\ref{sec:semantics}. We recall the definition below
for convenience: the $p_i'$s are defined as the unique path sequence such that there exists~$i_1<i_2<\cdots<i_{\ell+1}$ (delineating the boundaries of $p_1',\cdots,p_\ell'$) with~$i_1=1$, $i_{\ell+1}=n+1$ and satisfying the following.
\begin{align*}
    \forall k\in \set{1,\ldots,\ell}\quad &p_k' = p_{i_k}\cdots p_{i_{k+1}-1}
    \\
    \forall k\in\set{1,\ldots,\ell-1}\quad & \pathlen(p_{i_k}) \neq 0 \mathrel{\vee} \pathlen(p_{i_{k+1}}) \neq 0 
    \\
    \forall k\in \set{1,\ldots,\ell}\quad & \begin{cases}
    \text{either} & i_{k+1} = i_{k}+1 \text{ and } \pathlen(p_{i_k})\neq0\\
    \text{or}     & 
    \forall i~,i_k\leq i < i_{k+1},~\pathlen(p_{i})=0
    \end{cases}
\end{align*}

Note that~$\ell$ is at most~$2L+1$; that the number of paths~$\pi'_k$ such that~$\pathlen(\pi'_k)=0$ is at most~$L+1$; and that the number of paths~$\pi'_k$ such that~$\pathlen(\pi'_k)\neq 0$ is at most~$L$, in which case~$(i_{k+1}-i_k)=1$.
Since~$n>(L+1)(M+1)$, it follows that
\begin{align*}
    \sum_{\substack{k \\ 1\leq k < \ell\\\pathlen(p'_k)=0}} \kern-8pt 
    (i_{k+1}-i_k) = n -\kern-8pt  \sum_{\substack{k \\ 1\leq k < \ell\\\pathlen(p'_k)\neq0}}\kern-8pt 
    (i_{k+1}-i_k) >{}& (L+1)(M + 1) -L\\
    >{}& (L+1)M + 1
\end{align*}
hence there exists~$k$ such that~$\pathlen(p'_k)=0$ and~$(i_{k+1}-i_k)>M$.
From the definition of~$M$, it holds~$M \geq \card\setst*{\mu}{(p'_k,\mu)\in\sem{\pat}_G}$,
hence by the pigeon-hole principle, there are~$i,j$ such that~$i_k \leq i <j < i_{k+1}$ and~$\mu_i=\mu_j$.
It follows that
\begin{equation*}
    \bigcup_{\substack{x\\i_k\leq x < i_{k+1}}} \mu_x 
    =
    \bigcup_{\substack{x\\i_k\leq x < i_{k+1}\\x\neq j}} \mu_x~.
\end{equation*}
Note that the left part above is the definition of~$\mu_k'$, hence
\[
\collect((p_1,\mu_1),\ldots,(p_n,\mu_n))=
\collect((p_1,\mu_1),\ldots,(p_{j-1},\mu_{j-1}),(p_{j+1},\mu_{j+1}),\ldots,(p_n,\mu_n))
\]
or, in other words,~$(p,\mu)\in\sem{\pat}_G^{n-1}$.
\end{proof}

\section{Expressivity}
\label{sec:app-expressivity}

We prove that each regular query can be expressed in \gpmlplus. Because regular queries work in a simpler graph model in which there are no properties and undirected edges, we will need neither conditioned patterns nor undirected edge patterns. In fact, that data model also does not include node labels, but it is straightforward to include them, which we do. 

Consider a Datalog program defining a regular query. We first rewrite the program in such a way that that all user predicates except for the answer predicate $\ans$ are binary and are defined with rules whose bodies are connected. Towards this goal, we first eliminate all occurrences of user defined predicates in non transitive atoms, by substituting their definitions exhaustively. User predicates not used in any rule body should be also exhaustively removed. Now, each remaining non-answer user predicate is only used in transitive atoms and it remains to eliminate rules whose bodies are disconnected. Consider such a predicate that has among its defining rules a disconnected rule \[P(x_1,x_2) \colondash P_1(y_1, z_1), P_2(y_2, z_2), \dots, P_k(y_k, z_k)\] where $y_1, y_2, \dots y_k$ and $z_1, z_2, \dots z_k$ are variables from the set $\{x_1, x_2, \dots, x_m\}$ for some $m\geq 2$. Note that components that are disconnected from both $x_1$ and $x_2$ are essentially global Boolean side conditions: if we want to use the rule, we must check that they can be matched somewhere, but if we use the rule multiple times, there is no need to check it again. We shall now consider two cases, depending on whether $x_1$ and $x_2$ are in the same component or not. 

\begin{itemize}
\item Suppose first that $x_1$ and $x_2$ are in different connected components of the body of the rule above. Let us remove this rule from the program (keeping the other rules for $P$) and add a new rule \[\dot P(x_1,x_2) \colondash P_1(y_1, z_1), P_2(y_2, z_2), \dots, P_k(y_k, z_k)\] where $\dot P$ is a fresh predicate. If we now replace each occurrence of $P^{+}(x,y)$ with \[P^{+}(x,x'), \dot P(x',y'), P^{+}(y',y) \quad  \text{or} \quad  \dot P(x,y'), P^{+}(y',y) \quad  \text{or} \quad P^{+}(x,x'), \dot P(x',y) \quad  \text{or} \quad \dot P(x,y) \quad  \text{or} \quad P^{+}(x,y)\] and then eliminate $\dot P$ by substituting its definition, we obtain an equivalent program. The reason why it is equivalent is that when traversing the graph with $P$, there is no need to use the disconnected rule more than once. Indeed, whenever we have a sequence of nodes $u_1, u_2, \dots, u_t$ such that $P(u_i, u_{i+1})$ holds for all $i<t$ under the original definition of the predicate $P$, we also have that $P(u_i, u_{i+1})$ or $\dot P(u_i, u_{i+1})$ for all $i<t$ under the new definition of $P$. If the latter holds for more than one value of $i$, let $i_{\min}$ be the minimal and $i_{\max}$ the maximal of those values. We have $i_{\min} < i_{\max}$. It follows that $\dot P(u_{i_{\min}}, u_{i_{\max} +1})$ holds and we can remove $u_{i_{\min} + 1}, \dots, u_{i_{\max}}$ from the  sequence. Traverses that use the disconnected rule at most once are captured by the modified program.

\item The remaining possibility is that $x_1$ and $x_2$ are in the same connected component of the body of the rule above. This time we also eliminate the disconnected rule from the definition of $P$ (keeping the other rules for $P$), but additionally we define two new predicates, $\dot P$ and $\ddot P$. The definition of $\dot P$ is obtained from the original definition of $P$ by taking all rules for $P$ but replacing $P$ in the head with $\dot P$, except for the rule with disconnected components, in which we remove all connected components that contain neither $x_1$ nor $x_2$; these components are collected in a single rule defining a fresh predicate $\ddot P (z,z)$ where $z$ is an arbitrary variable used in any of these components. Now, each transitive atom $P^{+}(x,y)$ can be replaced with $P^{+}(x,y)$ or $\dot P^{+}(x,y) , \ddot P(z,z)$ where $z$ is a fresh variable, and $\ddot P$ can be eliminated by substituting its definition. The purpose of $P^{+}(x,y)$ is to maintain all derivations of $P$ that don't use the disconnected rule that we eliminated. 
The replacement with $\dot P^{+}(x,y) , \ddot P(z,z)$ is correct because when traversing the graph with $P$ it is enough to match the disconnected components of the considered rule once, as this match can be reused whenever the rule is applied. Indeed, consider a sequence of nodes $u_1, u_2, \dots, u_t$ such that $\dot P(u_i, u_{i+1})$ for all $i<t$ and $\ddot P(u)$ for some node $u$. If the modified rule is never used, we have $P(u_i, u_{i+1})$ for all $i<t$ under the old definition of $P$. If the modified rule is used, then we can extend each match of its body to the remaining connected components of the original rule's body by using the match of $\ddot P$ at node $u$.
\end{itemize}
This way we can eliminate all disconnected rules from the definitions of predicates used in transitive atoms. As the only remaining user predicate is the answer predicate, we are done. 

For technical convenience, let us further modify the resulting program so that for each user predicate $P$,
\begin{itemize}
    \item either $P$ is defined with a single rule of the from  \[P(x,x) \colondash A(x)\quad \text{or} \quad  P(x,y) \colondash a(x,y) \quad \text{or} \quad P(x,y) \colondash R^{+}(x,y)\] where $A$ is a node label, $a$ is an edge label, and $R$ is a user predicate, 
    \item or $P$ is defined by rules whose bodies are conjunctions of binary user predicates (without transitive closure). 
\end{itemize}
Note that this can be done without introducing rules with disconnected bodies, so it does not break the first stage of the preprocessing. We are now ready to construct an equivalent $\gpmlplus$ query. 

First, for each binary non-answer user predicate $P$ in the program, we construct by structural induction a \gpml pattern $\pi_P$ such that, for each graph $G$, we have that $P(u,u')$ holds in $G$ iff $(u,u') \in \sem{(x)\pi_P(y)}^{x,y}_G$ where $x$ and $y$ are fresh variables not used in $\pi_P$. The base cases are predicates defined with a single rule of the form \[P(x,x)\colondash A(x) \quad  \text{or} \quad P'(x,y)\colondash a(x,y)\,;\] for such predicates we use \[\pi_P = \nodelit{:A} \quad \text{and} \quad \pi_{P'}=\,\rightlit{:a}\,.\] For predicates defined with a single rule of the form \[P(x,y) \colondash R^{+}(x,y)\] we take the pattern \[\pi_P = (\pi_R)^{1..\infty}\,.\]
Finally, consider a non-answer predicate $P$ defined with $n$ rules using only binary user predicates (without transitive closure). We shall construct path patterns $\pi_{P,i}$ (using fresh variables) capturing the rules defining $P$ and then let \[\pi_P = \pi_{P,1} + \pi_{P,2} + \dots +\pi_{P,n}\,.\] 
Consider the $i$th rule defining $P$, say \[P(x_1,x_2) \colondash P_1(y_1,z_1), P_2(y_2,z_2),\dots, P_k(y_k,z_k)\] where $y_1, y_2, \dots y_k$ and $z_1, z_2, \dots z_k$ are variables from the set $\{x_1, x_2, \dots, x_m\}$ for some $m\geq 2$. Let $\pi_{P_1},\pi_{P_2}, \dots, \pi_{P_k}$ be the path patterns obtained inductively for predicates $P_1, P_2, \dots, P_k$, each using  fresh variables. We let \[\pi_{P,i} = (x_1) \,\big[\rightlit{} + \leftlit{}\big]^{0..\infty}\,
(y_1)\,\pi_{P_1}\,(z_1) \,\big[\rightlit{} + \leftlit{}\big]^{0..\infty}\, 
(y_2)\,\pi_{P_2}\,(z_2) \,\big[\rightlit{} + \leftlit{}\big]^{0..\infty}\,
\dots \,\big[\rightlit{} + \leftlit{}\big]^{0..\infty}\, 
(y_k)\,\pi_{P_k}\,(z_k) \,\big[\rightlit{} + \leftlit{}\big]^{0..\infty}\, 
(x_2)\,.\]
Because all rule bodies for non-answer predicates are connected, the body of the analyzed rule will always be matched in a connected fragment of the graph. Hence, the auxiliary patterns $\big[\rightlit{} + \leftlit{}\big]^{0..\infty}$ used to move from variable to variable do not affect the semantics. 

With the patterns $\pi_P$ at hand, we can translate the whole regular query. The answer predicate need not be binary and the bodies of the rules  defining it need not be connected, but we are at the very top of the Datalog program and we can simply use the join operator, followed by projection and union. 
Consider a rule defining the answer predicate, 
\[\ans(x_1,x_2,\dots, x_l) \colondash P_1(y_1,z_1), P_2(y_2,z_2),\dots, P_k(y_k,z_k)\] where $y_1, y_2, \dots, y_k$ and $z_1, z_2, \dots, z_k$ are variables from the set $\{z_1, z_2, \dots, z_m\}$ for some $m\geq l$.
We replace this rule with
\[\ans(x_1,x_2)\colondash\shortest\, (y_1)\pi_{P_1}(z_1), \shortest \, (y_2)\pi_{P_2}(z_2),\dots,\shortest \, (y_k)\pi_{P_k}(z_k)\]
where $\pi_{P_1}, \pi_{P_2}, \dots, \pi_{P_k}$ are the patterns (over fresh variables) obtained inductively for predicates $P_1, P_2, \dots, P_k$. Because we are only interested in the \emph{existence} of paths connecting $y_i$ to $z_i$, not the paths themselves, we can safely apply the restrictor $\shortest$. The resulting collection of rules constitutes a \gpmlplus query equivalent to the regular query defined by the considered Datalog program. 
\section{Complexity}
\label{sec:app-complexity}

By $|\pi|$ we denote the ``structural'' size of $\pi$, that is, the number of nodes in its parse tree plus the number of bits needed to represent numbers $n$ and $m$ in subexpressions of the form $\pi^{n..m}$. For a path $p$, we denote by $|p|$ the total number of node and edge ids in $p$. For a mapping $\mu$, we denote by $|\mu|$ the total length of paths occurring in $\mu$ plus the number of occurrences of variables in $\mu$.

\begin{lemma}
\label{lm:len}
Let $\rho \pi$ be a pattern with a restrictor on top. For a property graph $G$, if $(p,\mu)\in \sem{\rho \pi}_G$, then
\begin{enumerate}[(a)]
    \item $\pathlen(p)\leq |N|$, whenever $\rho$ is $\rsimple$;
    \item $\pathlen(p)\leq |E_d|+|E_u|$, whenever $\rho$ is $\rtrail$;
    \item $\pathlen(p)\leq (|N|+|E_d|+|E_u|)\times 2^{|\pi|}$, whenever $\rho$ is $\shortest$.
\end{enumerate}
\end{lemma}
\begin{proof}
    When $\rho$ is $\rsimple$, no node can be repeated, and thus $|p|\leq |N|$. Similarly, when $\rho$ is $\rtrail$, the path length is bounded by the total number of edges in the graph; namely, $|E_d|+|E_u|$.

    The case when $\rho$ is $\shortest$ can be obtained using standard automata-theoretic techniques. The challenge is how to deal with repetitions of the form $\pat\quantlit{{\interval{n}{m}}}$, where $n$ and $m$ are binary numbers. By unraveling the lower bounds $n$ into concatenations, such expressions can be turned into expressions of size $2^n$ that only have repetitions of the form $\pat\quantlit{{\interval{0}{\infty}}}$. (Notice that doing so requires special care of group variables, but this does not influence the length of $p$.) The bound $(|N|+|E_d|+|E_u|)\times 2^{|\pi|}$ then follows from a compilation of the expression into a finite automaton and the fact that a shortest path in the product of this automaton with the graph is also a shortest path in the graph itself. Given that the argument here is quite standard in the research literature (e.g.~\cite{B13,GeladeMN09,LosemannM-tods13,MendelzonW89}), we do not develop a fully formal automaton model for our queries.
\end{proof}

\begin{lemma}\label{l:bound-mu-size}
    Let~$G$ be a graph and $\pat$ be a well-typed pattern. For every path~$p$ of~$G$, the size $|\mu|$ of~$\mu$ such that~$(p,\mu)\in\sem{\pat}_G$ is at most $|p|\times (2^{|\pi|+1}-2)$.
    Here, $|p|$ denotes the number of occurrences of node and edge ids in $p$. 
\end{lemma}
\begin{proof} 
By induction over the structure of~$\pat$.  The lemma is obviously true for all base cases, and all inductive cases are easy except for a repeated pattern of the form~$\pat\quantlit{h..\infty}$, for some~$h$, where we have that
\begin{equation*}
    \sem{\pat\quantlit{h..\infty}}_G = \bigcup_{i=h}^{\infty} \sem{\pat}_G^i ~.
\end{equation*}

By definition of $\collect$, we have that every $(p,\mu) \in \sem{\pat\quantlit{h..\infty}}_G$ consists of a list  $(p_1,\mu_1),\ldots,(p_n,\mu_n)$, where $n \leq |p|$ and $p_1\cdots p_n = p$. Furthermore, we have that $\sum_{i=1}^n |p_i| \leq 2|p| -1$. If every $\mu_i$ has size at most $(2^{|\pi|+1}-2)$, then we have that the size of $(p_1,\mu_1),\ldots,(p_n,\mu_n)$ is bounded by 
\begin{align*}
  \sum_{i=1}^n|p_i| + \sum_{i=1}^n |\mu_i| & \leq \sum_{i=1}^n|p_i| + \sum_{i=1}^n |p_i|\times (2^{|\pi|+1}-2)\\
  & \leq \left(\sum_{i=1}^n|p_i|\right) (1+2^{|\pi|+1}-2)\\
  & \leq 2|p|\times(2^{|\pi|+1}-1)\\
  & \leq |p|\times(2^{|\pi|+2}-2)\\
  & \leq |p|\times (2^{|\pat\quantlit{h..\infty}|+1}-2)\\
\end{align*}
\end{proof}

\begin{lemma}\label{lem:check-without-vars}
Let $\pi$ be a pattern without variables and $G$ a property graph. Given a path $p$, we can compute the set of $(p',\mu)\in \sem{\pi}_G$ such that $p'$ is a subpath of $p$ in time polynomial in the size of $\pi$, $G$, and $p$. 
\end{lemma}
\begin{proof}
    Let us denote by $\sem{\pi}_G^p$ the set of $(p',\mu) \in \sem{\pi}_G$ such that $p'$ is a subpath of $p$. If $\pi$ does not have variables, then notice that all results $(p',\mu) \in \sem{\pi}_G$ have $\mu = \square$. Assume that $p = \pathval(u_0,e_0,u_1,\ldots,e_n,u_n)$. We prove by structural induction on $\pi$ that we can compute in polynomial time the set $\mathsf{Pairs(\pi')} := \{(i,j) \mid (\pathval(u_i,e_i,\ldots,u_j),\square) \in \sem{\pi'}_G^p\}$ for all subexpressions $\pi'$ of $\pi$.
    
    The base cases can clearly be computed in polynomial time. If $\pi'$ is of the form $\pi_1 \pi_2$, then $\mathsf{Pairs}(\pi')$ is the natural join of $\mathsf{Pairs}(\pi_1)$ with $\mathsf{Pairs}(\pi_2)$. If $\pi'$ is of the form $\pi_1 + \pi_2$, then $\mathsf{Pairs}(\pi')$ is the union of $\mathsf{Pairs}(\pi_1)$ with $\mathsf{Pairs}(\pi_2)$. If $\pi' = \pi_1^{n..m}$, then we can compute $\mathsf{Pairs}(\pi_1)^{n..m}$ in polynomial time using iterative squaring.
\end{proof}

\begin{lemma}
\label{lm:check}
Let $\pi$ be a fixed pattern and $G$ a property graph. Given a path $p$, we can enumerate the set of $(p',\mu) \in \sem{\pi}_G$ such that $p'$ is a subpath of $p$ in space polynomial in the size of $G$ and $p$.
\end{lemma}
\begin{proof}[Proof sketch]
    Let us denote by $\sem{\pi}_G^p$ the set of $(p',\mu) \in \sem{\pi}_G$ such that $p'$ is a subpath of $p$. We prove the lemma by structural induction on $\pi$. If $\pi$ does not have variables, then Lemma~\ref{lem:check-without-vars} tells us that we can compute $\sem{\pi}_G^p$ in polynomial time, even when $\pi$ is not fixed. So we assume that $\pi$ has at least one variable.
    
    The base cases are clear. For the inductive cases, we enumerate the relevant results $(p',\mu)$ of the subexpressions in polynomial space inductively and compose them to form a result for the entire expression.
    For concatenation, union, and conditioned patterns, we can do this by just following the inductive definition of $\sem{\pi_1\pi_2}_G$, $\sem{\pi_1 + \pi_2}_G$, and $\sem{\pi_{\langle \theta \rangle}}_G$. It remains to discuss repeated patterns, i.e., $\pi = \pi_1^{n..m}$. Since $\pi$ is constant, it suffices to consider the case $\pi = \pi_1^{n..n}$ and the case $\pi = \pi_1^{0..\infty}$. We first consider the case $\pi_1^{n..n}$. By the induction hypothesis, we can assume that we can enumerate  $\sem{\pi_1}_G^p$ using polynomial space. Since $\pi_1$ has at least one variable (we already dealt with the case without variables), we need to enumerate the set $(p,\mu)$ such that there exist $(p_1,\mu_1), \ldots, (p_n,\mu_n) \in \sem{\pi_1}_G^p$ where $p = p_1 \cdots p_n$ and $\mu = \collect [(p_1,\mu_1), \ldots, (p_n,\mu_n)]$. Since $n$ is a constant, we only need to keep constantly many $(p_i,\mu_i)$ in memory, so we can proceed in polynomial space by following the definition of $\collect$. We now consider the case $\pi = \pi_1^{0..\infty}$. In this case, we again follow the definition of $\collect$. Even though the number $n$ in a sequence $(p_1,\mu_1), \ldots, (p_n,\mu_n)$ with $(p_i,\mu_i) \in \sem{\pi_1}_G$ and $p = p_1 \cdots p_n$ can be arbitrarily large, $\collect$ always collapses those subpaths of length zero together. As such, only polynomially many different subpaths of $p$ need to be considered. For a given $(p,\mu)$, we therefore systematically enumerate all sequences $(p_1,\mu_1), \ldots, (p_n,\mu_n)$ with $p = p_1\cdots p_n$ and $(p_i,\mu_i) \in \sem{\pi_1}^p_G$. Since by Lemma~\ref{l:bound-mu-size}, the total size of the mappings $\mu_i$ is also polynomial in $|p|$ if $\pi$ is fixed, we can enumerate everything using polynomial space.
\end{proof}

\textsc{Theorem~\ref{th:compl}}. {\it The problem \textsc{Enumerate answers} can be solved by a Turing machine using exponential space (in $G$ and $Q$). If we consider the query $Q$ to be fixed (data complexity), then the machine uses only polynomial space.}
\begin{proof}[Proof sketch]
We first consider data complexity; namely, we assume $Q$ to be fixed. First we consider the query $Q = \rho \pi$, which is a single pattern with a restrictor on top. The Turing machine in this case operates by trying to enumerate all possible paths $p$ and mappings $\mu$ such that $(p,\mu)\in \sem{Q}_G$. W.l.o.g., we can assume an order on node and edge ids. 

Assume first that the query does not use the selector $\shortest$, but is of the form $\rho \pi$, with $\rho$ either $\rsimple$ or $\rtrail$. The machine starts enumerating all possible paths $p$ one by one in radix order, that is, in increasing length, and then by the ordering we assume on node and edge ids. By Lemma~\ref{lm:len}, we know that we only need to consider paths of polynomial length.
For each such path, we enumerate all possible $\mu$ that bind variables of $\pi$ to elements in $p$ one by one, which can be done in polynomial space in the size of $G$ and $p$ according to Lemma~\ref{lm:check}.
For a given path $p$, we can thus enumerate all possible answers $(p,\mu)$ using polynomial space. 
 
Assume now that the query uses the restrictor $\shortest$. In this case, we have to modify our machine slightly. Notice that in this case we might have the restrictor $\shortest$ either on top of $\rtrail$ or $\rsimple$, or just on its own. In all of these cases we proceed in the same manner: we enumerate the results of the query below $\shortest$, and the first time we encounter a result, we write it down to the output tape, \emph{and} store it to an additional tape that this machine has access to. When another path is then considered, in case that it could be the answer, its length is compared to the stored path, and eliminated in case its length is longer. The path enumeration procedure can be halted upon reaching length $|G|\times 2^{|\pi|}$  (Lemma~\ref{lm:len}).
 
 In a query that uses joins, the procedure can be thought as being nested for each part of the join. Given that the query $Q$ is fixed, the amount of work space used by the machine remains polynomial.
 
 Moving to combined complexity, we can observe that the space used thus far is exponential at worst. In fact, the main complication arises with iterations of the form  $\pi \quantlit{{\interval {n} {m}}}$, where $n$ and $m$ are concrete numbers, given that their representation (e.g. in binary) is exponentially more succinct than their magnitude, which dictates the number of graph elements that might need to be bound in group variables.
 \end{proof}

\textsc{Theorem}~\ref{theo:complexity-lower}
{\it The problem \textsc{Enumerate answers} cannot be solved by a Turing machine using polynomial amount of space (in $G$ and $Q$). }
\begin{proof}
  Queries of the form $x = \shortest \; \nodelit{\,} \,\rightlit{}\quantlit{\interval{k}{k}}\nodelit{
  \,}$ produce $2^k$ different paths on the graph $G$ with nodes $u,v$, with $a$-edges from $u$ to $v$ and back; and $b$-edges from $u$ to $v$ and back. Since $k$ is represented in binary, this number of paths is $\Theta(2^{2^n})$, where $n$ is the input size. A  polynomial space algorithm cannot represent the required $2^{2^n}$ configurations required to enumerate these paths without repetitions.
\end{proof}
\section{Extensions}
\label{sec:app-extensions}

\textsc{Proposition}~\ref{diophantine-prop}.
{\it  The data complexity of \gpml\ with arithmetic conditions is undecidable.}

\begin{proof}
We first show the undecidability of combined complexity and then extend the construction to show data complexity as well. We do it by reduction from the Diophantine equation (or Hilbert's 10th) problem, namely whether a multivariate polynomial $f(x_1,\ldots,x_m)$ with integer coefficients has a solution, i.e., numbers $v_1,\ldots,v_m\in\mathbb{N}$ so that $f(v_1,\ldots,v_m)=0$ (the problem is more commonly stated about finding solutions in $\mathbb Z$ but clearly this version is undecidable as well). 
Specifically we shall use a version of it where $m\leq 58$ and degrees of all the monomials in $f$ are bounded by $4$; this is already known to be undecidable \cite{jones-diophantine80}.

Create a graph $G_0$ with nodes $n_1,\cdots,n_{58}$, an $A_i$-labeled loop on every node $n_i$ for $i \leq 58$, and $A$-labeled edges from $n_i$ to $n_{i+1}$ for $i<58$. We assume that $n_1$ is labeled $S$ and has a value $\dv(n_1,k)=0$. Consider a pattern 
$$\pat_0 = \nodelit{u:S} \rightlit{x_1:A_1}\quantlit{\interval{0}{\infty}}\nodelit{} \rightlit{:A} \nodelit{} \rightlit{x_2:A_2}\quantlit{\interval{0}{\infty}}\nodelit{} \rightlit{:A}  \cdots \rightlit{:A}  \nodelit{} \rightlit{x_{58}:A_{58}}\quantlit{\interval{0}{\infty}}\nodelit{}$$
This pattern navigates along the $A$-edges, potentially looping several times over each $A_i$-labeled edge. Now consider $\pat_0\condlit{f(\cnt{x_1},\cdots,\cnt{x_{58}})=u.k}$. Then $G_0$ has a match for this pattern iff $f$ has a solution $v_1,\ldots,v_m\in\mathbb{N}$. Indeed, the match will loop $v_1$ times over $n_1$, then move to $n_2$, loop $v_2$ times, and so on. 

This shows undecidability of combined complexity as the polynomail $f$ is part of the query. We now extend the construction to data complexity. Let $\alpha$ range over the set $M$ of $5^{58}$ mappings from $\{1,\cdots,58\}$ to $\{0,\cdots,5\}$. By $m_\alpha$ we mean the monomial
$$x_1^{\alpha(1)} \cdot \cdots\cdot x_{58}^{\alpha(58)}$$
and we then assume without loss of generality that 
$$f = \sum_{\alpha\in M} c_\alpha\cdot m_\alpha$$
where $c_\alpha\in \mathbb{Z}$. Note that the degree of some of the monomials may be higher than $4$ but then the problem of the existence of integer solutions of $f$ is still undecidable, and furthermore the number of monomials is fixed. Also notice that since each degree is bounded by $4$, every monomial can be constructed as an arithmetic expression, e.g., we could write $u\cdot u\cdot w\cdot w\cdot w\cdot w$ instead of $u^2w^4$. 

We next extend $G_0$ to $G$ as follows. We enumerate the set $M$ as $\alpha_1,\cdots,\alpha_{|M|}$ and in $G_0$ continue the chain of $A$-labeled edges for another $M$ nodes which all will
have a property $\textit{coeff}$ holding values $c_{\alpha_1},\cdots,c_{\alpha_{|M|}}$. In addition, these nodes will have loops, labeled $B_1,\cdots,B_{|M|}$. 

Then, starting with
$\pat_0$ we proceed to define $\pat_i$ for $1 \leq i\leq M$ inductively as follows:

$$\pat_{i} = \big[\pat_{i-1} \rightlit{:A} \nodelit{} \rightlit{y_i:B_i}\quantlit{\interval{0}{\infty}}\nodelit{}\big]\condlit{\cnt{y_i}=y_i.\textit{coeff}\cdot m_1(\cnt{x_1},\cdots,\cnt{x_{58}}}$$

The effect of matching such a pattern is that the number of times $\cnt{y_{i}}$ the $B_i$-labeled loop is traversed equals the value of $c_{\alpha_i}\cdot m_{\alpha_i}$ over the values $\cnt{x_1},\cdots,\cnt{x_{58}}$. Indeed, as mentioned earlier, monomials $m_\alpha$ are proper arithmetic terms. 

Let $\pat=\pat_{|M|}$. Then $$\pat
\condlit{\cnt{y_1}+ \cdots + \cnt{y_{|M|}}=u.k}$$ matches iff $f(\cnt{x_1},\cdots,\cnt{x_{58}})=0$. Since now the pattern is fixed (all the information about $f$ is encoded in $G$), the data complexity is proved to be undecidable. 
\end{proof}

\end{document}